\documentclass{article}

\usepackage[paper=a4paper,margin=2.3cm,bottom=3cm]{geometry}
\usepackage{hyperref}
\makeatletter
\renewenvironment{table}
{\def\@floatboxreset{\reset@font\scriptsize\@setminipage}\@float{table}}
{\end@float}
\renewenvironment{table*}
{\def\@floatboxreset{\reset@font\scriptsize\@setminipage}\@dblfloat{table}}
{\end@dblfloat}
\makeatother
\usepackage{cite}
\usepackage{amsmath,amssymb,amsfonts}
\usepackage{algorithmic}
\usepackage{graphicx} %paquete para trabajar con imagenes
\usepackage{textcomp}
\usepackage{xcolor}

\usepackage{ifthen}
\input{dgfly_drawer.tex}

\usepackage{pgfplots}
\usepgfplotslibrary{statistics}%for boxplots
\pgfplotsset{compat=newest}
\pgfplotsset{minor grid style={dashed,very thin, color=blue!15}}
\pgfplotsset{major grid style={very thin, color=black!30}}

\pgfkeys{/pgf/number format/.cd,
fixed,
fixed zerofill=false,%<-- by default have no trailing 0s.
set thousands separator={\,},
1000 sep in fractionals,
sci E,%<-- 1.23E+1
}

\newenvironment{experimentfigurecomp}{\begin{figure*}}{\end{figure*}}
\newenvironment{experimentfigure}{\begin{figure}\scriptsize}{\end{figure}}

\makeatletter
\tikzset{nomorepostaction/.code=\let\tikz@postactions\pgfutil@empty}
\makeatother

\pgfkeys{
	/my/.is family,
	/my/mark/.estore in=\mymark,
	/my/mark=o,
	/my/region/.estore in=\myregion,
	/my/line/.estore in=\myline,
}

\newcommand{\kindseparator}{\hskip 0ex{}}

\pgfplotsset{
	automatically generated axis/.style={
		height=105pt,%may fit 3figures with 2 line caption
		width=145pt,%may fit 4 figures
		scaled ticks=false,
		xmin=0,
		xticklabel style={font=\tiny,/pgf/number format/.cd, fixed,/tikz/.cd},% formattin ticks' labels
		yticklabel style={font=\tiny,/pgf/number format/.cd, fixed,/tikz/.cd},% formattin ticks' labels
		x label style={at={(ticklabel cs:0.5, -5pt)},name={x label},anchor=north,font=\scriptsize},
		y label style={at={(ticklabel cs:0.5, -5pt)},name={y label},anchor=south,font=\scriptsize},
		enlargelimits=false,
		every axis title/.style={font=\tiny},
	},
	CT plot/.style={
		height=105pt,
		width=250pt,
		yticklabel={%
			\pgfkeys{/pgf/fpu}%
			\pgfmathparse{\tick}%
			\pgfmathfloatifflags{\pgfmathresult}{0}{0}{%
				\pgfmathprintnumber[sci,sci zerofill,precision=1]{\tick}%
			}%
			\pgfkeys{/pgf/fpu=false}%
		},
		every axis title/.style={font=\scriptsize},
		ymin=0,ymax=2e5,%
		scaled ticks=false,
		xticklabel style={font=\tiny,/pgf/number format/.cd,/tikz/.cd},% formattin ticks' labels
		yticklabel style={font=\tiny,/pgf/number format/.cd, fixed,precision=2,/tikz/.cd},% formattin ticks' labels
		x label style={at={(ticklabel cs:0.5, -5pt)},name={x label},anchor=north,font=\scriptsize},
		y label style={at={(ticklabel cs:0.5, -5pt)},name={y label},anchor=south,font=\scriptsize},
	},
	automatically generated symbolic/.style={
		height=105pt,
		width=250pt,
		xticklabel style={font=\tiny,text width=8ex,align=center},
		yticklabel style={font=\tiny,/pgf/number format/.cd, fixed,/tikz/.cd},% formattin ticks' labels
		x label style={at={(ticklabel cs:0.5, -5pt)},name={x label},anchor=north,font=\scriptsize},
		y label style={at={(ticklabel cs:0.5, -5pt)},name={y label},anchor=south,font=\scriptsize},
		every axis title/.style={font=\scriptsize},
	},
	first kind/.style={
		legend style={font=\scriptsize,fill=none},
		legend columns=7,legend cell align=left,
	},
	posterior kind/.style={
		legend style={draw=none},
	},
}
\tikzset{
	automatically generated plot/.style={
		/pgfplots/error bars/x dir=none,
		/pgfplots/error bars/y dir=none,
		/tikz/mark options={solid},
	},
	automatically generated bar plot/.style={
		/pgfplots/error bars/y dir=none,
	},
	automatically generated boxplot/.style={
		/pgfplots/boxplot={
			every average/.append style={/tikz/mark size=1.5pt,/tikz/mark=none},
			draw/lower whisker/.append code={
				\path (boxplot cs:\pgfplotsboxplotvalue{lower whisker}) ++(0pt,-5pt) coordinate (below whisker);
				\path (axis description cs: 0,0) ++(0pt,5pt) coordinate (mark origin);
				\path (below whisker |- mark origin) coordinate (mark position);
				\draw[dotted] (below whisker) -- (mark position);
				\node at (mark position) {\pgfuseplotmark{\mymark}};
			}
		},
	},
}

\def\xValiantxtext{Valiant}

\colorlet{shortest color}{olive}
\colorlet{Valiant color}{blue}
\colorlet{OMNI color}{red}
\colorlet{Polarized color}{violet}
\colorlet{Polarized UD color}{black}
\colorlet{DAL color}{teal}
\colorlet{OMNI SP color}{purple}

\tikzset{
	Polarized UD line/.style={solid,mark=square,color=Polarized UD color,thick},
	OMNI SP line/.style={solid,mark=triangle,color=OMNI SP color,thick},
	Polarized line/.style={dashed,mark=Mercedes star flipped,color=Polarized color},
	shortest line/.style={dashed,mark=o,color=shortest color},
	Valiant line/.style={dotted,mark=star,blue},
	OmniDim line/.style={dotted,mark=diamond,color=OMNI color},
	DAL line/.style={dotted,mark=pentagon,DAL color},
}
\tikzset{xshortestx/.style={automatically generated plot,shortest line}}
\tikzset{xValiantx/.style={automatically generated plot,Valiant line}}
\tikzset{xOmnixUPx/.style={automatically generated plot,OMNI SP line}}
\tikzset{xDALx/.style={automatically generated plot,DAL line}}
\tikzset{xPolarizedxwithxladderx/.style={automatically generated plot,Polarized line}}
\tikzset{xOmnixdimensionalxwithxRIIIxderoutesx/.style={automatically generated plot,OmniDim line}}
\tikzset{xPolarizedxUPx/.style={automatically generated plot,Polarized UD line}}

\tikzset{xPolarizedxwithxRRIRIIxescapexagRIIIxRVIRIVxxRIIIxRIVRVIIIxx/.style={automatically generated plot,black,solid,mark=square}}
\tikzset{xPolarizedxwithxRRIRIIxescapexagRIIxRVIRIVxxRIIIxRIVRVIIIxx/.style={automatically generated plot,green,dashed,mark=triangle}}
\tikzset{xPolarizedxwithxRRIRIIxescapexagRIxRVIRIVxxRIIIxRIVRVIIIxx/.style={automatically generated plot,blue,dotted,mark=o}}
\tikzset{xPolarizedxwithxRRIRIIxescapexagRIxRVIRIVxxRIxRIVRVIIIxx/.style={automatically generated plot,violet,solid,mark=diamond}}
\tikzset{xPolarizedxwithxRRIRIIxescapexagRIIxRVIRIVxxRIxRIVRVIIIxx/.style={automatically generated plot,red,dash dot,mark=star}}
\tikzset{xPolarizedxwithxRRIRIIxescapexagRIIIxRVIRIVxxRIxRIVRVIIIxx/.style={automatically generated plot,blue,dotted,mark=o}}
\tikzset{xPolarizedxwithxRRIRIIxescapexagRIVxRVIRIVxxRIxRIVRVIIIxx/.style={automatically generated plot,green,dashed,mark=triangle}}

\tikzset{
	shortest pattern/.style={pattern=north west lines},
	Valiant pattern/.style={pattern=grid},
	KSP UGAL pattern/.style={pattern=crosshatch},
	OmniDim pattern/.style={pattern=crosshatch},
	UGAL pattern/.style={pattern=crosshatch dots},
	Polarized pattern/.style={pattern=north east lines},
	Polarized SP pattern/.style={pattern=north east lines},
	OMNI SP pattern/.style={pattern=crosshatch},
}
\tikzset{
	Polarized SP bar/.style={automatically generated bar plot,fill=Polarized UD color!20,postaction={Polarized SP pattern}},
	OMNI SP bar/.style={automatically generated bar plot,fill=OMNI SP color!20,postaction={OMNI SP pattern}},
}
\tikzset{xshortestxbar/.style={automatically generated bar plot,fill=shortest color!20,postaction={shortest pattern},}}
\tikzset{xValiantxbar/.style={automatically generated bar plot,fill=Valiant color!20,postaction={Valiant pattern},}}
\tikzset{xOmnixUPxbar/.style={automatically generated bar plot,fill=KSP UGAL color!20,postaction={KSP UGAL pattern},}}
\tikzset{xDALxbar/.style={automatically generated bar plot,fill=UGAL color!20,postaction={UGAL pattern},}}
\tikzset{xPolarizedxwithxladderxbar/.style={automatically generated bar plot,fill=POL-ESCAPE-HCO color!20,postaction={Polarized pattern},}}
\tikzset{xOmnixdimensionalxwithxRIIIxderoutesxbar/.style={automatically generated bar plot,fill=OMNI color!20,postaction={OmniDim pattern},}}
\tikzset{xPolarizedxUPxbar/.style={automatically generated bar plot,fill=POL-ESCAPE-DOR color!20,postaction={Polarized pattern},}}

\tikzset{boxplotlegend/.style={
	/pgfplots/legend image code/.code={
		\draw (0,-4pt) rectangle ++(30pt,8pt);
		\draw[\myregion] (3pt,-3pt) rectangle ++(16pt,6pt);
		\draw[\myline,mark options={draw=black,solid}] plot coordinates {(25pt,0pt)};
	}
}}
\tikzset{xshortestxboxplotlegend/.style={
	/my/region=xshortestxbar,
	/my/line=shortest line,
	boxplotlegend,
}}
\tikzset{xValiantxboxplotlegend/.style={
	/my/region=xValiantxbar,
	/my/line=Valiant line,
	boxplotlegend,
}}
\tikzset{xRVIIIxKSPxrandomizedxUGALboxplotlegend/.style={
	/my/region=xOmnixUPxbar,
	/my/line=KSP UGAL line,
	boxplotlegend,
}}
\tikzset{xDALxboxplotlegend/.style={
	/my/region=xDALxbar,
	/my/line=UGAL line,
	boxplotlegend,
}}
\tikzset{xPolarizedxwithxladderxboxplotlegend/.style={
	/my/region=xPolarizedxwithxladderxbar,
	/my/line=Polarized line,
	boxplotlegend,
}}
\tikzset{xOmnixdimensionalxwithxRIIIxderoutesxboxplotlegend/.style={
	/my/region=xOmnixdimensionalxwithxRIIIxderoutesxbar,
	/my/line=OmniDim line,
	boxplotlegend,
}}
\tikzset{xPolarizedxUPxboxplotlegend/.style={
	/my/region=xPolarizedxUPxbar,
	/my/line=Polarized line,
	boxplotlegend,
}}

\tikzset{xshortestxboxplot/.style={
	automatically generated boxplot,fill=shortest color!20,
	/my/mark=o,
	every path/.style={postaction={nomorepostaction,shortest pattern}},
}}
\tikzset{xValiantxboxplot/.style={
	automatically generated boxplot,fill=Valiant color!20,
	/my/mark=star,
	every path/.style={postaction={nomorepostaction,Valiant pattern},}
}}
\tikzset{xOmnixUPxboxplot/.style={
	automatically generated boxplot,fill=KSP UGAL color!20,
	/my/mark=diamond,
	every path/.style={postaction={nomorepostaction,KSP UGAL pattern},}
}}
\tikzset{xDALxboxplot/.style={
	automatically generated boxplot,fill=UGAL color!20,
	/my/mark=pentagon,
	every path/.style={postaction={nomorepostaction,UGAL pattern},}
}}
\tikzset{xPolarizedxwithxladderxboxplot/.style={
	automatically generated boxplot,fill=POL-ESCAPE-HCO color!20,
	/my/mark=square,
	every path/.style={postaction={nomorepostaction,Polarized pattern},}
}}
\tikzset{xOmnixdimensionalxwithxRIIIxderoutesxboxplot/.style={
	automatically generated boxplot,fill=OMNI color!20,
	/my/mark=diamond,
	every path/.style={postaction={nomorepostaction,OmniDim pattern},}
}}
\tikzset{xPolarizedxUPxboxplot/.style={
	automatically generated boxplot,fill=POL-ESCAPE-DOR color!20,
	/my/mark=square,
	every path/.style={postaction={nomorepostaction,Polarized pattern},}
}}

\tikzset{xuniformxbar/.style={automatically generated bar plot,fill=red!20,postaction={pattern=horizontal lines},}}
\tikzset{xrandomxserverxpermutationxbar/.style={automatically generated bar plot,fill=green!20,postaction={pattern=grid},}}
\tikzset{xgloballyxshuffledxdestinationsxbar/.style={automatically generated bar plot,fill=blue!20,postaction={pattern=crosshatch},}}
\tikzset{xdimensionxcomplementxreversexbar/.style={automatically generated bar plot,fill=black!20,postaction={pattern=dots},}}
\tikzset{xxRIstxstaticxxblockxpermutationxtoxneighbourxbar/.style={automatically generated bar plot,fill=violet!20,postaction={pattern=north east lines},}}
\tikzset{xDCRxRIIDxbar/.style={automatically generated bar plot,fill=black!20,postaction={pattern=dots},}}
\tikzset{xRIVxRIVxblockxpermutationxtoxneighbourxbar/.style={automatically generated bar plot,fill=violet!20,postaction={pattern=north east lines},}}
\tikzset{xRIIIxRIVxblockxxxcolumnxtoxneighbourxbar/.style={automatically generated bar plot,fill=orange!20,postaction={pattern=vertical lines},}}

\pgfplotsset{
	load/.style={
		ymin=0,ymax=1
	},
	delay/.style={
		ymin=0,ymax=500,
	},
	hops/.style={},
	jain/.style={
		ymin=0.85,ymax=1,%
	}
}

\tikzset{uniform/.style={automatically generated plot,black,solid,mark=*}}
\tikzset{randomxserverxpermutation/.style={automatically generated plot,red,solid,mark=square*}}
\tikzset{dimensionxcomplementxreversex/.style={automatically generated plot,violet,solid,mark=triangle*}}
\tikzset{regularpermutationtoneighbour/.style={automatically generated plot,blue,solid,mark=diamond*}}
\tikzset{xuniformx/.style={uniform}}
\tikzset{xrandomxserverxpermutationx/.style={randomxserverxpermutation}}
\tikzset{xDCRxRIIDx/.style={dimensionxcomplementxreversex}}

\tikzset{xOmnixRRIRIIxwithxUpDownStarxcrossbothxprioxbar/.style={OMNI SP bar}}
\tikzset{xPolarizedxRRIRIIxwithxUpDownStarxcrossbothxprioxbar/.style={Polarized SP bar}}

\tikzset{xOmnixUpDownStarxcrossbothxpriox/.style={OMNI SP line}}

\def\xxxxBrinrxrandomxxxxxtext{bRINR}
\def\xxxxUgalxxxAlexxlabelxxxxxtext{sRINR}
\def\xxxxomnixxxxxtext{Omni-WAR}
\def\xxxxValiantxxxxxtext{MIN/VLB}
\def\xxxxEmbeddedxHyperXxRIIDxsourcexxxxxtext{TERA-HX2}
\def\xxxxEmbeddedxHyperXxRIIIDxxxxxtext{TERA-HX3}
\def\xxxxUgalxxxValiantxxxxxtext{UGAL}

\def\xxxxEmbeddedxHyperXxRIIDxxxxxtext{TERA-HX2}

\def\xomnixtext{Omni-WAR}
\def\xUgalxxxValiantxtext{UGAL}

\def\xNaturalxtxt{Natural}
\def\xRandomxtxt{Random}
\def\xxxxDORxSxHypercubexxxxxtext{DOR-TERA-HX3}%%Es un 2x2x2
\def\xxxxOmnixWARxxxxxtext{Omni-WAR}
\def\xxxxDimWARxxxxxtext{DimWAR}
\def\xxxxORIxturnxSxHypercubexxxxxtext{O1TURN-TERA-HX3}%%Es un 2x2x2

\def\xEmbeddedxHyperXxRIIDxsourcextext{TERA-HX2}
\def\xMeshxopportunisticxtext{TERA-PATH}

\def\TERATREEtext{TERA-4-Tree}

\def\TERAHXIIItext{TERA-HX3}
\def\UGALtext{UGAL}

\colorlet{Valiant}{blue}
\colorlet{TERA HX2}{red}
\colorlet{TERA HX3}{violet}
\colorlet{TERA path}{cyan}
\colorlet{TERA 4tree}{brown}
\colorlet{UGAL}{olive}
\colorlet{sRINR}{teal}
\colorlet{bRINR}{lime}
\colorlet{OmniWAR}{green!80!black}
\colorlet{DimWAR}{purple}
\colorlet{O1TURN TERA HX3}{pink}
\tikzset{
	Valiant/.style={color=Valiant,mark=star},
	TERA HX2/.style={color=TERA HX2,mark=o},
	TERA HX3/.style={color=TERA HX3,mark=square},
	TERA path/.style={color=TERA path,mark=triangle},
	TERA 4tree/.style={color=TERA 4tree,mark=Mercedes star flipped},
	UGAL/.style={color=UGAL,mark=triangle},
	sRINR/.style={color=sRINR,mark=Mercedes star flipped},
	OmniWAR/.style={color=OmniWAR,mark=diamond},
	bRINR bar/.style={fill=bRINR!20,postaction={pattern=horizontal lines}},
	sRINR bar/.style={fill=sRINR!20,postaction={pattern=grid}},
	UGAL bar/.style={fill=UGAL!20,postaction={pattern=crosshatch}},
	Valiant bar/.style={fill=Valiant!20,postaction={pattern=north east lines}},
	OmniWAR bar/.style={fill=OmniWAR!20,postaction={pattern=horizontal lines}},
	TERA HX2 bar/.style={fill=TERA HX2!20,postaction={pattern=north east lines}},
	TERA HX3 bar/.style={fill=TERA HX3!20,postaction={pattern=grid}},
	DimWAR bar/.style={fill=DimWAR!20,postaction={pattern=north east lines}},
	O1TURN TERA HX3 bar/.style={fill=O1TURN TERA HX3!20,postaction={pattern=crosshatch}},
}

\tikzset{xMeshxopportunisticx/.style={TERA path}}

\tikzset{xEmbeddedxHyperXxRIIDxsourcex/.style={TERA HX2}}

\tikzset{xValiantx/.style={Valiant}}

\tikzset{xUgalValiantx/.style={UGAL}}

\tikzset{xRIVxtreexopportunisticx/.style={TERA 4tree}}

\tikzset{hamingthreedimensions/.style={TERA HX3}}

\tikzset{xxxxUgalxxxValiantxxxxx/.style={UGAL}}

\tikzset{xxxxUgalxxxAlexxlabelxxxxx/.style={sRINR}}

\tikzset{xxxxEmbeddedxHyperXxRIIDxsourcexxxxx/.style={TERA HX2}}

\tikzset{xxxxomnixxxxx/.style={OmniWAR}}

\tikzset{xxxxValiantxxxxx/.style={Valiant}}

\tikzset{xxxxEmbeddedxHyperXxRIIIDxxxxx/.style={TERA HX3}}

\tikzset{xxxxBrinrxrandomxxxxxbar/.style={bRINR bar}}
\tikzset{xxxxUgalxxxAlexxlabelxxxxxbar/.style={sRINR bar}}
\tikzset{xxxxUgalxxxValiantxxxxxbar/.style={UGAL bar}}

\tikzset{xxxxValiantxxxxxbar/.style={Valiant bar}}

\tikzset{xxxxomnixxxxxbar/.style={OmniWAR bar}}
\tikzset{xxxxEmbeddedxHyperXxRIIDxxxxxbar/.style={TERA HX2 bar}}
\tikzset{xxxxEmbeddedxHyperXxRIIIDxxxxxbar/.style={TERA HX3 bar}}

\tikzset{xxxxOmnixWARxxxxxbar/.style={OmniWAR bar}}
\tikzset{xxxxDimWARxxxxxbar/.style={DimWAR bar}}
\tikzset{xxxxORIxturnxSxHypercubexxxxxbar/.style={O1TURN TERA HX3 bar}}
\tikzset{xxxxDORxSxHypercubexxxxxbar/.style={TERA HX3 bar}}

\usepackage{caption}  % Required for \captionof
\usepackage{algorithm2e}

\usepackage[utf8]{inputenc} %paquete para usar teclado espanol
\usepackage{csquotes} % recommended
\usepackage[T1]{fontenc}%to be able to hyphenate accents
\usepackage[english]{babel} %
\usepackage[protrusion=true,expansion=true]{microtype}
\usepackage{hyperref} %hiper-referencias, claro!
\usepackage{multirow}
\usepackage{amsthm}%To have a proof environment.

\newtheorem{theorem}{Theorem}[section]
\newtheorem{definition}[theorem]{Definition}

\newtheorem{claim}[theorem]{Claim}

\usepackage{tikz}
\usetikzlibrary{calc,trees,fadings,fit,positioning,arrows.meta,shapes.symbols,decorations.markings,shapes.geometric,patterns,decorations.pathreplacing}
\usepackage[utf8]{inputenc}% will fix Günter/García at Mon??
\usepackage{subfig}

\usetikzlibrary{positioning}%relative positions.
\usetikzlibrary{decorations.text}%Decorate with text along a path.
\usetikzlibrary{calc}%advanced calculation of positions.
\usetikzlibrary{arrows, decorations.markings}%Decorate with arrow marks.
\usetikzlibrary{matrix}%Table matrices in tikz.

\usetikzlibrary{patterns}%Decorate bars and boxplots.
\usetikzlibrary{external}%Serialize some figures onto standalone PDFs.
\tikzexternaldisable

\usepackage{pgfplots}
\usepackage{colortbl} %\rowcolor
\usepackage{listings}
\usepackage{ragged2e}%for \justify
\usepackage{mathtools}
\begin{document}

    \title{Deadlock-free routing for Full-mesh networks without using Virtual Channels \\
    {\large postprint version}\\
    {\small A. Cano, C. Camarero, C. Martínez, and R. Beivide, “Deadlock free routing for full-mesh networks without using virtual channels,” in HOTI25: 32nd IEEE Hot Interconnects symposium. IEEE, 2025. }\\
    {\small\url{https://doi.org/10.1109/HOTI66940.2025.00020}}\\
    }

    \author{Alejandro Cano \and Cristóbal Camarero \and Carmen Martínez \and Ramón Beivide \and \\
    {\small \{alejandro.cano, cristobal.camarero, carmen.martinez, ramon.beivide\}@unican.es}
    }

    \date{}
    \maketitle

    \begin{abstract}
        High-radix, low-diameter networks like HyperX and Dragonfly use a Full-mesh core, and rely on multiple virtual channels (VCs) to avoid packet deadlocks in adaptive routing.
        However, VCs introduce significant overhead in the switch in terms of area, power, and design complexity, limiting the switch scalability.
        This paper starts by revisiting VC-less routing through link ordering schemes in Full-mesh networks, which offer implementation simplicity but suffer from performance degradation under adversarial traffic.
        Thus, to overcome these challenges, we propose \textbf{TERA} (Topology-Embedded Routing Algorithm), a novel routing algorithm which employs an embedded physical subnetwork to provide deadlock-free non-minimal paths without using VCs.

        In a Full-mesh network, TERA outperforms link ordering routing algorithms by 80\% when dealing with adversarial traffic, and up to 100\% in application kernels.
        Furthermore, compared to other VC-based approaches, it reduces buffer requirements by 50\%, while maintaining comparable latency and throughput.
        Lastly, early results from a 2D-HyperX evaluation show that TERA outperforms state-of-the-art algorithms that use the same number of VCs, achieving performance improvements of up to 32\%.
    \end{abstract}

    \section{Introduction}
   High-radix, low-diameter networks are the foundation of many modern interconnect architectures in supercomputers and data centers.
   Topologies such as Dragonfly~\cite{Kim_dgfly_ISCA} and HyperX~\cite{HyperX, kim_flat_CAL} rely on a Full-mesh (FM) core.
   %Dragonfly+\cite{dragonfly_plus, megafly}
   In FM networks every switch is directly connected to every other switch, or said equivalently, the graph of the network is a complete graph, which has diameter 1.
   In Figure~\ref{fig:fm} a FM with four switches and four servers per switch is represented.
%   Dragonfly and HyperX, together with Dragonfly+\cite{dragonfly_plus, megafly}, also use a complete graph for their global networks.
  % which provides minimal path lengths, high bisection bandwidth, and excellent scalability.

  % The graph of the switches in a Full-mesh is a Complete Graph, having
    \emph{Minimal routing} (MIN) in a FM—which involves sending packets directly from source to destination—does not introduce cyclic buffer dependencies and is inherently deadlock-free.
    Furthermore, MIN delivers excellent performance under uniform traffic loads, on which packets are evenly distributed across all network paths.

    However, under adversarial or bursty traffic loads, MIN can lead to significant load imbalance and performance degradation.
    In such scenarios, some links become congested while others remain idle.
    To enhance overall throughput, non-minimal paths are employed, detouring packets through intermediate switches to improve network utilization.
    In FM networks, non-minimal paths typically consist of two hops: from the source to an intermediate switch, and from the intermediate switch to the destination.
    These paths are usually based on the Valiant Load-Balancing (VLB) scheme~\cite{Valiant_ACM}, which routes each packet through a randomly selected intermediate switch before forwarding it to its destination.

    \begin{figure}
        \centering
        \includegraphics[scale=0.6]{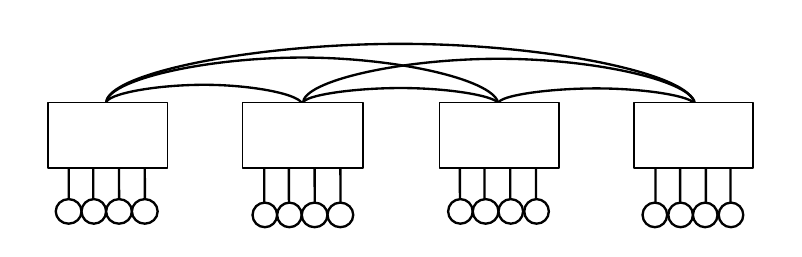}
        \caption{Full-mesh network with 16 servers: 4 switches with 4 servers per switch.}
        \label{fig:fm}
    \end{figure}

    To improve performance, adaptive routing algorithms are responsible for selecting between minimal and non-minimal paths based on network conditions, such as buffer occupancy or link utilization.
    However, unlike MIN paths, non-minimal routes in a FM introduce cyclic dependencies, potentially leading to packet deadlocks.
    Therefore, routing in FM networks must support deadlock-free operation for any mix of MIN and non-MIN routes.
    Two common strategies to avoid routing deadlocks in FM networks are:
    \begin{itemize}
        \item \textbf{Virtual Channels} (or \emph{buffer restriction}).
        This approach involves having multiple virtual channels (VCs or FIFO buffers) within each switch port, allowing packets to avoid cycles by traversing different VCs under a certain order.
        Some proposals include Virtual Ordered Buffered Classes~\cite{Gunter, Merlin}, or escape channels~\cite{Duato_general}.
        \item \textbf{Link ordering} (or \emph{path restriction}).
        This method assigns labels to links in the network, ensuring that packets follow a specific order of the links that avoids deadlocks~\cite{BoomGate, turnmodel}.
    \end{itemize}

    The first approach is effective and capable of achieving optimal performance.
    However, it requires additional buffering resources, which occupy a significant portion of the switch area and increase power consumption as well as the complexity of the switch design.
    Adaptive routing algorithms for FM-based networks, such as UGAL~\cite{Singh} and Omni-WAR~\cite{Kim_omni} exemplify this class of solutions.

    The second approach avoids the need for VCs and is simpler to implement.
    Nevertheless, it can result in uneven path utilization and performance degradation under adverse traffic conditions.
    Routing schemes such as bRINR~\cite{BoomGate} and Up*/Down*~\cite{autonet} are representative of this method.

    In this paper, we first examine the link ordering approach in Section~\ref{sec:ordering}, exploring its limitations when applying to FM networks.
    We prove that link ordering schemes impose inherent performance bounds, motivating the need for an alternative strategy.
    To this end, we propose TERA (Topology-Embedded Routing Algorithm), a novel mechanism that enables deadlock-free non-minimal routing in a FM without relying on VCs, while maintaining high performance.
    TERA is evaluated against state-of-the-art routings for FM networks, including Omni-WAR~\cite{Kim_omni}, a well-established scheme known for its high throughput and low latency.

    The paper is organized as follows.
    In Section~\ref{sec:motivation}, the motivation behind our work is discussed.
    Section~\ref{sec:ordering} is devoted to exploring link orderings, showing that these schemes have inherent limitations.
    Section~\ref{sec:TERA} introduces TERA, a VC-less deadlock-free routing algorithm for Full-mesh networks.
    Using the different scenarios described in Section~\ref{sec:methodology}, TERA is evaluated and compared against other algorithms in Section~\ref{sec:results}.
 % In Section~\ref{sec:discussion}, we will discuss the implications of our findings to other networks based on FMs.
    Finally, in Section~\ref{sec:conclusion}, the paper contributions are summarized.

    \section{Motivation}\label{sec:motivation}

    \subsection{Buffering resources}
    Buffers dominate the switch area and take away an important part of the power budget~\cite{Scott,dai2017scalable,BoomGate,bufferconsumption}.
    Modern high-performance switches in supercomputer and data center networks require extensive buffering resources, presenting two challenges, as described next.

    \subsubsection{Buffer Size}
   % In addition to the number of buffers, it is also relevant their size.
    Several factors condition the buffer depth required in a switch.
    First, the Round-Trip Time (RTT) delay must be considered.
    The buffer size should be enough to cover the RTT latency to have a continuous transmission.
    It is directly related to the bandwidth of the switch-ports and to the distances covered by the wires. %, which can be up to 100m in an HPC center.
    In addition, sufficient buffer depth is necessary to absorb bursts of traffic~\cite{slingshot}.

    On top of that, common solutions must accommodate enough VCs in each network port to guaranty deadlock-free routing and, in some cases, to avoid protocol deadlock as well as to support multiple Quality-of-Service (QoS) levels.

    Conventional deadlock-avoidance mechanisms for packet routing—such as those proposed in~\cite{Gunter, Merlin, Duato_general}—typically require a minimum of two virtual channels when applied to FM topologies.
    This requirement scales with the number of QoS levels supported by the network.
    For example, in a system offering 16 QoS levels and using 2 VCs per level to ensure deadlock-free adaptive routing, each port would require a total of 32 buffers.
    Furthermore, if protocol-level deadlock must also be prevented, the VC count must be doubled, resulting in 64 VCs per port.
%    \textbf{REVISAR ESTO}
    %A general expression for the number of buffers per port is $l \times 2 \times d$, where $l$ is the number of QoS levels, and $d$ is the %number of buffers needed to avoid deadlock-free routing.

 %   \begin{itemize}
%        \item The Round-Trip Time (RTT) latency. The buffer size should be enough to cover the RTT latency to have a continuous transmission. It is directly related to the bandwidth of the switch-ports, and to the distances covered by the wires, which can be up to 100m in an HPC center [].%The %minimum buffer size in bits is given by:
%%        \begin{equation}
%%            B_{\text{min}} = \text{Bw} \times \frac{2d}{v}
%%            \label{eq:buffer_size}
%%        \end{equation}
%%        where $\text{Bw}$ is bandwidth (Gbps), $d$ is inter-router distance (m), and $v = 2 \times 10^8$ m/s (optical signal velocity).
%
%        \item The need to absorb bursts of traffic \textbf{[Aqui algo]}.
%    \end{itemize}

  %  \subsubsection{Number of buffers per port}

\subsubsection{Buffer Utilization Inefficiency}

     Common deadlock avoidance schemes for adaptive routing algorithms, such as UGAL~\cite{Singh} and Omni-WAR~\cite{Kim_omni}, rely on adding buffers to support minimal and non-minimal paths.
    These algorithms require at least 2 VCs per port to be deadlock-free in a FM\@.
    The first VC is used for minimal or non-minimal hops, while the second is used only for non-minimal hops.
    Consequently, when only minimal paths are employed, merely half of the available buffers are utilized.
    In general for any topology, the highest-order virtual channels remain mostly unused, while the lowest-ordered buffers present a high occupancy.
    This results in a significant waste of resources.

    Furthermore, prior work~\cite{canosbacpad} has demonstrated that injecting MIN packets on VCs other than the lowest-ordered ones introduces network instability.
    Thus, there is no clear way to fully utilize all available buffers.

    \subsection{High-radix low-diameter networks}

    %The Full-mesh serves as the foundational topology for many high-radix, low-diameter interconnects used in large-scale systems.
    Notable topologies, such as Dragonfly, HyperX, and Dragonfly+ \cite{dragonfly_plus}  \cite{megafly}, utilize FMs at the group or dimension level.
    A Dragonfly is employed in the current first three top-ranked supercomputers: El Capitan, Frontier, and Aurora.
    The three of them rely on complete graph intra-group and inter-group connectivity.
    Other supercomputers in the Top500~\cite{Top500} also implement these class of topologies.
    %Therefore, improving routing algorithms in FMs can directly benefit the performance and efficiency of these state-of-the-art systems.

    The routing mechanisms employed in these networks demand a high number of VCs to support deadlock-free routing.
    Typical requirements include 4 VCs in a 2D-HyperX, 6 VCs in a 3D-HyperX, and 4/2 VCs in Dragonfly for local/global ports.
    This stems from the fact that routing in a FM is not inherently deadlock-free when non-minimal paths are used.
    Thus, avoiding the necessity of VCs for deadlock-free routing in the FM core, directly translates into lower buffer requirements for the larger topologies built upon it.

    Importantly, some network technologies, such as InfiniBand~\cite{schneider2016ensuring}, support adaptive routing but do not permit in-transit VC shifting.
    As a result, deadlock-free routing schemes that do not rely on VCs are particularly desirable in these contexts.

 \section{Link ordering Schemes without VCs}
    \label{sec:ordering}

    In link ordering schemes, each directed link (arc) is assigned a number or \textsl{label}.
    Any valid path must follow a strictly increasing sequence of labels, which prevents deadlock.
    In a FM, depending on the selected ordering, the number of usable non-minimal paths can vary, as well as the proportion of such paths between pairs of switches.
    As we will see next, there is a trade-off between maximizing the number of non-minimal paths and ensuring a fair distribution of them between pairs of switches.
    This inherent property complicates the applicability of the method and suggests the search of routing mechanisms based on other principles.

    From now on, we will formally refer to the Complete graph, the underlying topology of the Full-mesh network.
    \begin{definition}
		The \textsl{Complete Graph} over the set of vertices $V=\{0,1,\dotsc,n-1\}$, denoted as $K_n$, is the graph where every pair of distinct vertices is connected.
		This is, the set of edges is $E=\{ \{a,b\} \mid a,b\in V,\ a\neq b \}$.
    \end{definition}

	Then, its number of links is $m=|E|=\binom{n}{2}=\frac{n(n-1)}{2}$. %We will also work over the set of directed edges, or \textsl{arcs}, which we denote by $A$ and has size $|A|=2m=n(n-1)$.
    From a source node to a destination node, there is 1 minimal path of length 1 and $n-2$ paths of length 2.
    As there are $n(n-1)$ pairs of nodes, there are a total of $n(n-1)(n-2)$ non-minimal paths in the graph.

 %The simplest case comprises three switches, which we can see as a directed 3-cycle. In this cycle, each switch sends its packets choosing the %previous switch as destination, and routing those packets through the other switch as intermediate.
%One possibility to avoid this deadlock is to limit what routes can be taken.
%To maximize adaptiveness, the number of allowed paths should be the maximum possible.

Allowing the use of all paths of length 2 for routing without VCs results in deadlock.
To avoid it, in~\cite{BoomGate}, the authors limit these 2-hop paths by using bRINR (\textsl{balanced} Restricted Intermediate-node Non-minimal Routing), aimed at balancing intermediate nodes while maximizing path diversity.
%The baseline routing which is proposed is RINR, which is equivalent to the Up*/Down* routing algorithm~\ref{Autonet}.
This routing algorithm attains the maximum number of possible non‐minimal paths for any link ordering scheme $\tfrac{2}{3}n(n-1)(n-2)$ but suffers from uneven path distribution.
It also ensures that each pair of switches has at least 2 intermediates.
However, some links are overloaded because they are used by many source/destination pairs.
This causes bottlenecks and hotspots in specific switches.

Before proposing a new link ordering scheme, let us show a result that proves that pursuing equalized utilization of all the links is a limiting factor in the number of possible non-minimal paths.
The proof is included in Appendix~\ref{app:link_ordering}.

    \begin{theorem}
		If a routing scheme for a FM is based on an ordering of the arcs and ensures that all links can be used by the same number of source/destination pairs, then the number of allowed paths of length 2 is $\frac{1}{2}n(n-1)(n-2)$.

        %Ensuring all the links have the same utilization, the number of non-minimal deadlock-free paths is $\frac{1}{2}\times(n \times (n-1) \times (n-2))$.
        \label{claim:utilizationfairpaths}
    \end{theorem}

This proves the trade-off between opportunities for adaptiveness and load-balancing.
The bRINR scheme~\cite{BoomGate} maximizes the number of possible intermediate nodes, on average.
However, as such value is above the expressed in the previous theorem, this implies that some arcs are used more than others.
As we will see, this imbalance manifests as a performance degradation for some traffic patterns.
This is due to unfairness issues which grow with the network size.

Then, it is natural to look for a different scheme looking for a perfect balance of non-minimal routes between switches by sacrificing the number of selectable paths.
This trade-off is similar to the one between the Turn Model and DOR routing~\cite{Dally}, where the former offers more paths but the latter achieves a fairer path distribution and higher overall performance.
Next, we introduce sRINR (\textsl{symmetric} RINR), a new link ordering that prioritizes a fair distribution of paths instead of attempting to maximize their number.

%Thus, in some cases, it is more important to ensure a fair distribution of paths among switches than to maximize the total number of paths, to prevent the apparition of congestion.

%\subsection{Full-Mesh network topology}

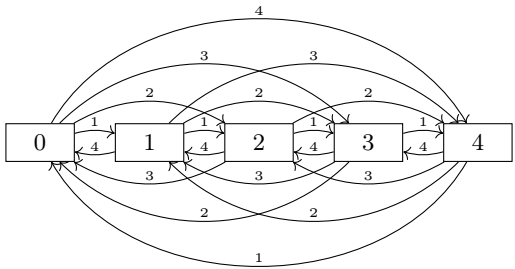
\begin{figure}
    \centering%
    \def\hud{\mbox{}\hskip 2ex plus 2ex minus 2ex\mbox{}}%
    \begin{tikzpicture}[
        scale=1.2,
        switch/.style={draw,inner sep=1.5pt,font=\small,minimum height=0.5cm,minimum width=0.9cm},
    %switchcomm/.style={blue,draw,inner sep=3pt,font=\small,minimum height=0.6cm,minimum width=1.1cm},
    %server/.style={draw,circle,inner sep=.5pt,font=\small},
%        every edge/.append style={overlay},
    ]
        \def\N{4}
        \pgfmathtruncatemacro\pN{\N-1}

        \foreach \y in {0,...,\N}
            {
            \node[switch] (nodo\y) at (1.2*\y, 0) {\y};
        }
        \foreach \xa in {0,...,\N}
            {
            \foreach \port in {1,...,\N}
                {
                \pgfmathtruncatemacro\sumswitch{\xa+\port}
                \pgfmathtruncatemacro\xb{Mod(\sumswitch, \N+1)}
                \ifthenelse{ \sumswitch > \N }{
                    \pgfmathtruncatemacro\angle{(\xb-\xa)*15}
                    \draw (nodo\xb) edge[out=\angle,in=180-\angle, <-] node[above, yshift=-0.2em] {\tiny \port} (nodo\xa);
                }{
                    \pgfmathtruncatemacro\angle{(\xb-\xa)*15}
                    \draw (nodo\xa) edge[out=\angle,in=180-\angle, ->] node[above, yshift=-0.2em] {\tiny \port} (nodo\xb);
                }
            }
        }
    \end{tikzpicture}%
    \caption{K$_4$ with sRINR labelling.}
    \label{fig:labelledCG}%
\end{figure}

%    A new link ordering, denoted as sRINR (\textsl{symmetric} RINR), is introduced in this paper.
%It ensures fairness by slightly sacrificing path diversity, improving the performance of bRINR.
%The next definition specifies the ordering.
	%which increment the unfairness in the network.
    %However, Claim~\ref{claim:utilizationfairpaths} doesn't assure how those paths are distributed among switches.

    \begin{definition}\label{def:sRINR}
        The sRINR ordering scheme assigns a number to an arc based on the distance between the switches it connects.
        For two switches $i$ and $j$ ($i \neq j$) connected by an edge, the arc from $i$ to $j$ is assigned the number $D(i,j)\in [0,n-1]$ with $D(i, j) \equiv (j - i) \mod{n}$.
        %This function assigns a distance from $+1$ to $+(n-1)$ for an ordered pair of switches.
    \end{definition}

Figure~\ref{fig:labelledCG} represents a K$_4$ labelled in this way.
The following claim bounds the number of intermediates.
The proof is included in Appendix~\ref{app:link_ordering}.

    \begin{claim}
       In the sRINR ordering, the minimum number of allowed intermediates for a source/destination pair is $\frac{n-4}{2}$.
        \label{claim:intermediates_forward}
    \end{claim}

For $n\geq 8$ this number is higher than the minimum number of intermediates provided by bRINR.

%Next, two claims are stated showing the best and worst case performance for sRINR.
%sRINR has a best case scenario, matching Valiant performance, and a worst case scenario which halves the accepted load.
%This is because there are links shared by two switches at the same time.
%Hence, it can be proved that the performance for shift patterns in sRINR is 0.5 flits/cycle/server, which matches Valiant routing performance.
%The worst case performance in sRINR is 0.25 flits/cycle/server, which halves Valiant routing performance.

Along with this theoretical analysis, Section~\ref{subsec:link_ordering} provides an evaluation of both sRINR and bRINR.
The results show that sRINR clearly outperforms bRINR.
Nevertheless, its overall performance remains below state-of-the-art routing algorithms based on VCs, indicating that link ordering schemes still impose strict limitations.

%[Aqui se dice que no es suficiente esto, y que hay que hacer algo mas desde otra perspectiva, como lo de TERA]

%Carmen: miramos luego donde escajar esto, pero no es un routing que se vaya a comparar con la propuesta
%Nue routing~\cite{nuerouting} takes an intermediate approach between TERA and link ordering, but Nue is a source routing algorithm which only embeds trees in the topology, which has a lot of limitations.
%Furthermore, it needs more than 1VC to work properly.[Algo mas tengo que poner, me falta detalle del paper]

    \section{The Topology Embedded Routing Algorithm}
    \label{sec:TERA}

%        \begin{itemize}
%            \item Provide an equal distribution of paths among switches.
%            \item Allow every MIN path and a large portion of non-minimal paths. And furthermore, in a way to avoid pathological wort-case permutations.
%            \item Fulfill a low and constant bound on the maximum number of hops.
%        \end{itemize}

%To, the next section introduces a routing algorithm that surpasses the performance bounds of the link ordering approaches.
    In this section we introduce TERA, a deadlock-free routing algorithm conceived to address the limitations of the link ordering schemes.
    Note that a complete graph has a very rich connectivity, and TERA takes advantage of this fact by breaking the topology into two physical parts.
    The main part should contain most of the links and be rich on allowed routes, so it can be freely used to route packets.
    The remaining part must correspond to a topology that serves in avoiding deadlock without using VCs.
    Let us define more formally the two topologies composing TERA.

%
%    A complete graph has every other graph as subgraph and  takes advantage of this fact.
%    There are multiple topologies that have routing mechanism that avoid deadlock without using VCs, for example, trees using MIN routing.
%    Looking for a high-performance, deadlock-free, and VC-absent routing mechanism for FM we propose to break the topology into two parts.
%    The main part of the topology should contain most of the links and be rich on allowed routes.
%    The remaining part must correspond to a topology that helps in avoiding deadlock.
%    Their combination faces the challenge of fairness and bounding the length of the routes.
%
%%%CARMEN: Esta terminología de grafos creo que puede despistar
%    The \textbf{Topology Embedded Routing Algorithm (TERA)} operates on the Full Mesh topology \( FM_n \) using a spanning subgraph and its complement.
%    The spanning subgraph provides a set of deadlock free paths, or as we called them \textbf{service paths}, while the complement graph of the spanning subgraph is freely used to route packets.
%    Let us introduce the following definition of the two subgraphs used in the routing algorithm.

    \begin{definition}[Main and service topologies] \label{def:subgraph}
        For a FM$_n$, a \textsl{service topology} $S$ is an embedded spanning topology using a deadlock-free minimal routing algorithm.
        By \textsl{main topology} $M$ we refer to the topology with the remaining links.
        A link is said a \textsl{service} link if it belongs to the service topology.
    \end{definition}

    Note that the service topology must span all the $n$ switches from the FM, as it has to be able to route from any source to any destination.
    There are different topologies that provide a deadlock-free set of minimal paths. %do not need VCs to be deadlock-free but must use minimal routing schemes that restrict the use of paths to avoid cyclic dependencies.
    This is the case, for example, of trees using up/down routing or meshes and hypercubes using DOR~\cite{Dally}.
 %   Rings and tori are also deadlock-free if Bubble flow control is employed~\cite{Puente_bubble}. %And in general, Up*/Down* routing can be used to reach the destination switch in a deadlock-free way [REF].

  	In Figure~\ref{fig:mesh_embedded}, two FM$_4$ are illustrated: one with a 2-tree (or 1D-mesh) service topology and the other with a Hypercube (or 2D-mesh in this case).
    Service links are represented with solid lines, while main links are depicted with dashed lines.
    Note that any path between switches through service links for the embedded 1D-Mesh is minimal, the same as the DOR paths for the hypercube.

        \begin{figure}
    \centering%
    \def\hud{\mbox{}\hskip 4ex plus 4ex minus 4ex\mbox{}}%
    \begin{tikzpicture}[
        scale=1.4,
        switch/.style={draw,inner sep=1.5pt,font=\small,minimum height=0.5cm,minimum width=0.9cm},
        %switchcomm/.style={blue,draw,inner sep=3pt,font=\small,minimum height=0.6cm,minimum width=1.1cm},
        %server/.style={draw,circle,inner sep=.5pt,font=\small},
        every edge/.append style={overlay},
		safe/.style={very thick},
		complement/.style={very thin,dashed},
    ]
        \node[switch] (nodo00) at (0, 0) {0};
        \node[switch] (nodo01) at (0, 1) {2};
        \node[switch] (nodo10) at (1, 0) {1};
        \node[switch] (nodo11) at (1, 1) {3};
        %x links
        \draw (nodo00) edge[safe,out=0,in=-180] (nodo10);
        \draw (nodo01) edge[safe,out=0,in=-180] (nodo11);
        %y links
        \draw (nodo00) edge[complement,out=90,in=-90] (nodo01);
        \draw (nodo10) edge[complement,out=90,in=-90] (nodo11);
        %Cross links
        \draw (nodo00) edge[complement,out=45,in=-135] (nodo11);
        \draw (nodo01) edge[safe, out=-45,in=135] (nodo10);
        \node at (0.5, -0.5) {Path};
    \end{tikzpicture}%
    \hud
    \begin{tikzpicture}[
        scale=1.4,
        switch/.style={draw,inner sep=1.5pt,font=\small,minimum height=0.5cm,minimum width=0.9cm},
        every edge/.append style={overlay},
		safe/.style={very thick},
		complement/.style={very thin,dashed},
    ]
        \node[switch] (nodo00) at (0, 0) {0};
        \node[switch] (nodo01) at (0, 1) {2};
        \node[switch] (nodo10) at (1, 0) {1};
        \node[switch] (nodo11) at (1, 1) {3};
        %x links
        \draw (nodo00) edge[safe,out=0,in=-180] (nodo10);
        \draw (nodo01) edge[safe,out=0,in=-180] (nodo11);
        %y links
        \draw (nodo00) edge[safe,out=90,in=-90] (nodo01);
        \draw (nodo10) edge[safe,out=90,in=-90] (nodo11);
        %Cross links
        \draw (nodo00) edge[complement,out=45,in=-135] (nodo11);
        \draw (nodo01) edge[complement,out=-45,in=135] (nodo10);
        \node at (0.5, -0.5) {Hypercube};
    \end{tikzpicture}%
    \hud
    \caption{$K_4$ with a Path, and a Hypercube embedded.}%
    \label{fig:mesh_embedded}%
\end{figure}
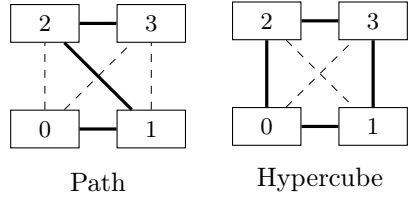

        TERA is described in Algorithm~\ref{alg:two}.
        The following notation is used in the pseudocode:

    \begin{itemize}
        \item The set $R_{\mathrm{main}}(x)$ contains the indexes of the ports of switch $x$ that belong to the main topology.
        \item The set $R_{\mathrm{serv}}(x, y)$ contains the indexes of the ports in switch $x$ that follow a service path to $y$.
        \item The set $R_{\mathrm{min}}(x, y)$ contains the index of the port in a switch $x$ that connects to the switch $y$ as a single element.
%        \item The boolean $b_{\mathrm{serv}}$ is true if the packet is in a service link.
%        This is used to actually bound path length and avoid the possibility of deadlock.
%        Note that without this flag packets should eventually reach the destination.
%        Moreover, it is unnecessary when the Service topology has diameter 2, e.g. being a 2D-HyperX.
    \end{itemize}

%%%CARMEN: analizar para que se dice esto, y escribirlo en el sitio correspondiente

%    The set $R_{min}(a, b)$ contains as single element the index of the port in switch $a$ that connects to switch $b$.
%	Note that since the link from $a$ to $b$ es either in $E(C)$ or in $E(S)$.
%    When it is in $E(C)$, it is clear that $R_{min}\in R_{comp}(a)$. When the link belongs to $E(S)$, we can assume without loss of generality that $R_{min}(a)\in R_s(a,b)$, as otherwise the path it could be introduced without causing deadlock. Thus, we conclude that $R_{min}(a) \in R_s(a,b) \cup R_{comp}(a)$.

    As it can be seen, TERA always evaluates the option of taking a service and a MIN path.
    But, if the packet is at the injection ports of the source switch, it also includes the option of taking any port from the main topology as a non-minimal hop.

    \SetKwComment{Comment}{/* }{ */}
    \RestyleAlgo{ruled}

    \begin{algorithm}
        %\caption{Routing Algorithm with Embedded Topology}\label{alg:two}
        \caption{TERA}\label{alg:two}
        \KwIn{Current switch: $current$, destination switch: $destination$}
        \KwOut{port}
        $ports \gets R_{\mathrm{serv}}(current, destination)$

        \uIf{\text{packet is at an injection port}}{
            $ports \gets ports \cup R_{\mathrm{main}}(current)$
        }
        \Else{
            $ports \gets ports \cup R_{\mathrm{min}}(current, destination)$
        }

        $candidates \gets \emptyset$

        \For{each $p$ in $ports$}{
            \uIf{$p$ connects to $destination$}{
                Insert $(p, occupancy[p])$ into $candidates$ \Comment{\small The occupancy is the weight}
            }
            \Else{
                Insert $(p, occupancy[p] + q)$ into $candidates$ \Comment{\small Penalize non-minimal paths}
            }
        }

        \Return the port in $(port,weight)\in candidates$ with the minimum $weight$. Ties are broken randomly.
    \end{algorithm}

%    \begin{theorem}
        %If $R_s$ over $S = (V, E')$ is deadlock-free, the Graph-Embedded Routing Algorithm is deadlock-free.
%		TERA is deadlock-free.
%    \end{theorem}
 %   \begin{proof} Let us assume, by way of contradiction, that there is a deadlock using TERA. This means that every packet is blocked, with all its possible next buffer candidates provided by the routing also blocked. In particular,.
 %       \end{proof}

In the algorithm, $q$ denotes a penalty applied to non-minimal paths.
It will be set in Section~\ref{sec:methodology}

TERA is deadlock-free.
This can be clearly understood by noting that every packet always has a valid routing option along a service path, and packets on these paths can always make forward progress.
Consequently, if a packet can't advance in the main network, sufficient buffer space will eventually free up in the service path, allowing it to proceed and preventing deadlock.

Furthermore, TERA is also livelock-free, and the maximum number of hops a packet can do is 1 plus the diameter of the service topology.
%as once a packet has done hop, it can only stay in the service topology following a service path, or use a minimal hop to the destination.
%It is worth noting that if the service topology has a diameter of 2—meaning all Service paths are of length 2, as in the case of the 2D-HyperX—the boolean variable $b_{\text{serv}}$ becomes unnecessary.

%With respect to livelock, ...

%    In comparison with other link ordering schemes, we are only restricting some links in the topology to be used in order to avoid deadlock.
%    In Table~\ref{tab:complete_intermediates} is the difference of non-minimal paths available.
%
%
%    Safe paths can cause the hops of a packet to increment, depending on the diameter of the embedded topology.
%    For instance, in the topology in Figure~\ref{fig:hamming_simple_live} if a server in Switch 0 sends a packet to switch 3, a valid route could be: 0->1->2->3 of 4 hops, while Polarized at most would do 2: 0->1->3.
%    Therefore, this path extension can alter the route length, and also can introduce livelock issues which are going to be discussed in the next sections.

%    \input{../figures/cgraph_embedded_simple_livelock}
%    \input{../figures/cgraph_embedded_livelock}

    \subsection{Evaluation of service topologies}\label{sub:safe-graph}

   The candidates for the service topology have to fulfill four key criteria:

    \begin{itemize}
        \item \textbf{Deadlock-Free Minimal Routing}: The topology must have a deadlock-free minimal routing algorithm that does not require VCs.
        \item \textbf{Edge and Vertex Symmetry}: A symmetric topology achieves a balanced performance in the network. Asymmetries can lead to congestion \cite{CamareroTC}.
        \item \textbf{Low Diameter}: Long service paths allow for worst-case scenarios, as they can spread congestion throughout the whole service topology.
        \item \textbf{Bounded Number of Links}: A high number of links imposes more restrictions on non-minimal paths, reducing the path diversity of the main topology.
    \end{itemize}

    \begin{table}
        \centering
        \begin{tabular}{|c|c|c|c|c|c|c|}
            \hline
            \textbf{Topology} & \textbf{Sym.} & \textbf{Diameter} & \textbf{$\#$Links} & \textbf{Routing} \\
            \hline
            $d$-Mesh & No & $O(n^{\frac{1}{d}})$ & $O(dn)$  & DOR\\
            \hline
           % $d$-Torus & Yes & $O(n^{\frac{1}{d}})$ & $O(dn)$  & Bubble R.\\
            %\hline
            $k$-Tree & No & $O(\log_k n)$ & $O(n)$ & Up/Down \\
            \hline
            Hypercube & Yes & $O(\log_2 n)$ & $O(n\log n)$ & DOR \\
            \hline
            $d$-HyperX & Yes & $O(d)$ & $O(dn^{1+\frac{1}{d}})$ & DOR \\
            \hline
        \end{tabular}
        \caption{Properties of different service topologies.}\label{tab:comparison_safe_topologies}
    \end{table}

    Table~\ref{tab:comparison_safe_topologies} compares these properties across four topology families.
    Mesh topologies have low degree, which entails a small number of links.
    This provides the main topology a greater connectivity and higher non-minimal bandwidth.
    However, their major drawback is a high diameter, which grows linearly with $n$.
%This can lead to worst-case scenarios under high congestion, where long service paths propagate congestion throughout the service topology, %ultimately degrading performance.

    Another option are $k$-trees, which reduce the diameter to $O(\log_k n)$ and employ a low number of links.
    However, they are not symmetric: the root becomes a central point of contention, and leaf nodes typically have only one upward link, creating bottlenecks.

    Lastly, the HyperX topology is a notable option, as it is a symmetric low-diameter topology.
    In particular, the 2D-HyperX and 3D-HyperX are considered the most suitable candidates.
    The 2D-HyperX topology has a diameter of 2, which is the lowest possible for an embedded topology.
    %However, HyperX topologies have a large degree, which reduces the number of links available for the main topology, limiting non-minimal path diversity.
    %Nevertheless, as the size of the FM increases, the proportion of links used for the service topology decreases, making it a suitable candidate.

    %%Carmen: mirar si hace falta hablar del hipercubo, pero al menos sacarlo de la parte de la HyperX
    % Lastly, the HyperX topology is notable option for consideration, as it a symmetric low-diameter topology.
%    In particular, the 2D-HyperX, 3D-HyperX, and Hypercube are considered the most suitable candidates.
%    The 2D-HyperX topology has a diameter of 2, which is the lowest possible for an embedded topology.
%    The 3D-HyperX has a diameter of 3, while the Hypercube achieves a logarithmic diameter of $\log n$, similar to $k$-trees, but with the added benefit of being symmetric.
%    The main drawback of HyperX topologies is their high number of links, which reduces the number of links available for the complement graph, limiting non-minimal path diversity.
%    However, as there is not a significant difference in throughput between the Hypercube and HyperX-3D, the Hypercube will be omitted to lighten the exposition. [ESTO HAY QUE PONERLO UN POCO MEJOR]
%    However, as the size of the FM increases, the proportion of links used for the safe topology decreases, leading to improved throughput.

	Using a service topology with a high degree, such as HyperX, reduces the number of links available for the main topology, limiting non-minimal path diversity.
	To better grasp the impact of this, let us estimate the throughput obtained by TERA with different service topologies under adverse random switch permutation traffic.
	This is shown in Figure~\ref{fig:safe_links_ratio}.
	Specifically, the curves follow the expression $\frac{1}{1 + p^{-1}}$, where $p$ is the degree of the main topology divided by $n-1$.
This is proved in Appendix \ref{app:TERA}.
For small FM sizes, the differences in performance among service topologies can be notable.
However, as the FM size increases, the curves converge.
In Section~\ref{subsec:seleccion}, the analysis demonstrates that the HyperX topology, with its optimal combination of low diameter and symmetry, represents an ideal choice.

    \begin{figure}
        \centering
        \includegraphics[width=0.5\textwidth]{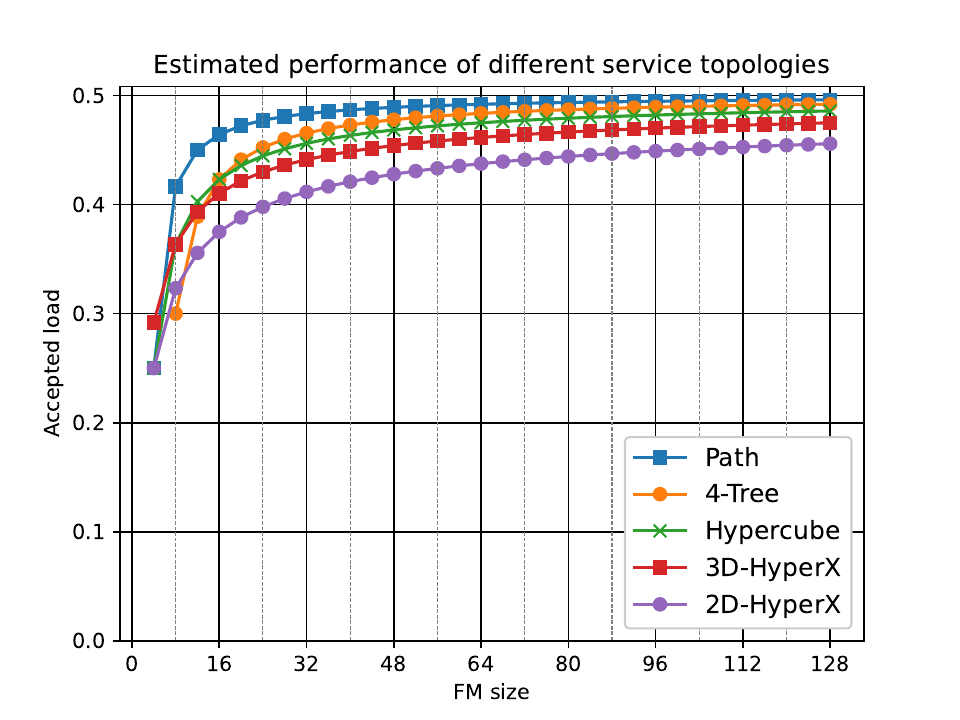}
        \caption{Estimated throughput for different service topologies.}
        \label{fig:safe_links_ratio}
    \end{figure}

    \section{Methodology}
    \label{sec:methodology}

    To compare the performance of the routing mechanisms studied in this work, we use the CAMINOS~\cite{CAMINOS,CAMARERO2025105136} network simulator.
    It is an event-driven simulator that models the switch microarchitecture and operates at flit level.
    %Greedy allocation is employed~\cite{kim2006adaptive}, wherein each packet at an input port independently selects a single output port, with collisions resolved randomly.
    The switch operates with a $2\times$ speedup and a random allocator.
    It provides buffer space for 10 packets per virtual channel at the input ports and 5 packets per virtual channel at the output ports.
    Each packet consists of 16 flits~\footnote{In \hbox{\url{https://github.com/alexcano98/TERA-routing-HOTI-2025-reproducibility}} full detail of our experimental setup, including the TERA implementation and simulator version is available.}

    The evaluated routings are MIN, bRINR, sRINR, TERA, Omni-WAR, UGAL and Valiant.
    TERA will use different service topologies.
    This will be indicated as a suffix in the acronym; for example, TERA-HX2, means TERA using a 2D-HyperX as service topology.

    With respect to the TERA pseudocode in Algorithm~\ref{alg:two}, after an experimental sweep, 54 was determined as a suitable penalty $q$ for the parameters used in our experiments.
    This value implies a penalty similar to the existence of slightly more than 3 packets in the buffer of any non-minimal path.

     Only one FIFO or VC is required by MIN, bRINR, sRINR, and TERA, whereas Omni-WAR, UGAL, and Valiant require 2 VCs to ensure deadlock-freedom, which doubles the buffering space and increases the design complexity of the switch.
     MIN and Valiant are only used as baselines.
%    bRINR, sRINR and TERA use only one FIFO or VC while Omni-WAR, UGAL and Valiant need 2 VCs to be deadlock-free, doubling the buffering space and complicating its management, as VCs must be multiplexed and de-multiplexed over the switch ports.

    A FM$_{64}$ with 4096 servers is extensively simulated.
    Nevertheless, to demonstrate the applicability of TERA in larger high-radix low-diameter networks, an initial evaluation of an $8 \times 8$ 2D-HyperX network is     included.

    Three modes of synthetic traffic generation are used:

    \begin{itemize}
        \item \textbf{Fixed generation}: Each server generates a total number of 1250 packets to be sent following a specific traffic pattern.
        The time to consume all packets is shown.
        It will serve to evaluate/discard the link ordering and the TERA routing mechanisms.
        \item \textbf{Bernoulli generation}: All servers generate traffic continuously during 80K cycles at a given rate of offered load, following a specific admissible traffic pattern.
        The considered performance metrics are explained below.
        \item \textbf{Application kernels}: A synthetic communication kernel is simulated in the whole FM, and the time to completion is measured.
        Processes of the kernel will be assigned to servers using both linear and random mapping.
    \end{itemize}

    For Bernoulli traffic, the metrics presented include average accepted throughput, average message latency, hop distribution, and Jain index for load generation.
    The hop distribution will be shown for the maximum offered load, while the other metrics will be displayed for different injected loads.
    The first three metrics are commonly used in the technical literature, while the last one provides a more specific view of the network fairness.
    The Jain index~\cite{jain1984quantitative} takes into account the load generated by all servers in the network.
    It is calculated as $\frac{ \left( \sum_{i=1}^n x_i \right)^2 } {n \sum_{i=1}^n x^2_i }$,
    where $x_i$ is the load generated by the server $i$ and $n$ is the total number of servers.
    A Jain index of $1.0$ indicates perfect equity, with all servers injecting the same load.
    Lower values suggest disparity between servers, the lower the worse.

    The admissible synthetic traffic patterns employed are:
    \begin{itemize}
        \item \textbf{Uniform} (\textbf{UN}): Each server generates messages destined to a random server of the network.
        This pattern could model different scenarios in a data center, where the distribution of traffic is uniform across the network due to the application (all-to-all communications in machine learning, for example), or due to a random mapping of tasks.
        \item \textbf{Random switch permutation} (\textbf{RSP}): It represents any permutation of switches.
        In the RSP, all the servers of a switch generate traffic to the servers of a destination switch.
        The map source to destination switch is a random permutation of the $n$ switches.
        \item \textbf{Fixed random} (\textbf{FR})\cite{BGQ_hood}: Each server selects a random server from the network to send a message, potentially creating endpoint bottlenecks.
        %It has been utilized in.
        \item \textbf{Switch Cartesian transforms}: The servers from switch $x$ send traffic to the servers in switch $f(x)$, where $f$ is a function that transforms the index of the switch.
        We call \textbf{shift} to the function $f(x)=x+1$ and \textbf{complement} to $f(x)=-x-1$.
    \end{itemize}

    The following application kernels are considered:
    \begin{itemize}
        \item \textbf{All2All}: Classical send loop.
        In iteration $i$ task $t$ sends to $t+i$~\cite{Thakur}.
        \item \textbf{Stencil 2D}: All the processes are arranged in a 2D grid, and each process communicates with the 8 processes in its Moore neighborhood.
        \item \textbf{Stencil 3D}: All the processes are arranged in a 3D grid, and each process communicates with its 26 neighbours across faces, edges, and corners.
        \item \textbf{FFT 3D}: A FFT3D with pencil decompositions using a 2D grid of processes and each process manages a pencil of data.
        Communications arise from partial transpositions of data and are All2All across each row or column~\cite{orozco2012demystifying}.
        \item \textbf{All-reduce}: The Rabenseifner algorithm, composed of a scatter-reduce and an all-gather.
        The algorithm is optimal in terms of bandwidth and number of messages for power of two processes~\cite{rabenseifner2004optimization}.
    \end{itemize}

    %The application kernels will be showed for a FM$_{64}$.

    \section{Performance Evaluation}
    \label{sec:results}

	Several empirical results are shown next.
	Subsection~\ref{subsec:link_ordering} is devoted to evaluating the link ordering mechanisms.
	Subsection~\ref{subsec:seleccion} study the differences among possible service topologies.
	Subsections~\ref{sec:static} and~\ref{subsec:apps} evaluate TERA with a HyperX service topology in different scenarios.
    Finally, Subsection~\ref{sec:discussion} evaluates the behaviour of TERA in a two-dimensional topology.

    \subsection{Link ordering evaluation}\label{subsec:link_ordering}

    \bgroup

%% %% -- common pgfplots prelude --
%% %\newenvironment{experimentfigure}{\begin{figure}[H]\tikzexternalenable}{\tikzexternaldisable\end{figure}}
%% %\newenvironment{experimentfigure}{\begin{figure*}}{\end{figure*}}
%% %\newcommand{\kindseparator}{\hskip 0ex{}}
%% \pgfplotsset{compat=newest}
%% \makeatletter
%% %required for fill plus pattern on boxplot
%% \tikzset{nomorepostaction/.code=\let\tikz@postactions\pgfutil@empty}
%% \makeatother
%% \pgfplotsset{minor grid style={dashed,very thin, color=blue!15}}
%% \pgfplotsset{major grid style={very thin, color=black!30}}
\pgfplotsset{
	automatically generated axis/.style={
		%default: height=207pt, width=240pt. 240:207 ~~ 7:6
		%height=115pt,%may fit 3figures with 1 line caption
		height=105pt,%may fit 3figures with 2 line caption
		width=174pt,
		scaled ticks=false,
		xticklabel style={font=\tiny,/pgf/number format/.cd, fixed,/tikz/.cd},% formattin ticks' labels
		yticklabel style={font=\tiny,/pgf/number format/.cd, fixed,/tikz/.cd},% formattin ticks' labels
		x label style={at={(ticklabel cs:0.5, -5pt)},name={x label},anchor=north,font=\scriptsize},
		y label style={at={(ticklabel cs:0.5, -5pt)},name={y label},anchor=south,font=\scriptsize},
	},
	automatically generated symbolic/.style={
		height=90pt,%120pt,%105pt,
		width=100pt,
		xticklabel style={font=\tiny}, %,rotate=90
		yticklabel style={font=\tiny,/pgf/number format/.cd, fixed,/tikz/.cd},% formattin ticks' labels
		x label style={at={(ticklabel cs:0.5, -5pt)},name={x label},anchor=north,font=\scriptsize},
		y label style={at={(ticklabel cs:0.5, -5pt)},name={y label},anchor=south,font=\scriptsize},
%		axis y line*=left,
	},
	first kind/.style={
		%The first axis on each line of plots
		%legend style={overlay,at={(0.50,1.05)},anchor=south,font=\scriptsize,fill=none},
		%legend style={at={(0.00,1.01)},anchor=south west,font=\scriptsize,fill=none},
		%legend style={at={($(axis description cs:0.00,1.01)!(current page.center)!(axis description cs:1.00,1.01)$)},anchor=south,font=\scriptsize,fill=none},
		legend style={font=\scriptsize,fill=none},
		legend columns=4,legend cell align=left,
	},
	posterior kind/.style={
		%Axis following the first on each line of plots
		%legend style={at={(0.50,1.05)},overlay,anchor=south,font=\tiny,fill=none},
		legend style={draw=none},
	},
}
\newcommand\captionprologue{X: }
\newcommand\experimenttitle{03-burst-facil/cycles.pdf (all 48 done)}
\newcommand\experimentheader{\tiny 03-burst-facil:cycles.pdf (all 48 done)\\pdflatex on \today\\version=heads/alex-stable-dirty-520894dfe13e6f413f0827db9d73d75acb82f89f(0.6.3)}
%\tikzset{xxxxBrinrxrandomxxxxx/.style={automatically generated plot,red,solid,mark=o}}
%\tikzset{xxxxBrinrxrandomxxxxxbar/.style={automatically generated bar plot,fill=red!20,postaction={pattern=horizontal lines},}}
%\tikzset{xxxxBrinrxrandomxxxxxboxplot/.style={automatically generated boxplot,fill=red!20,every path/.style={postaction={nomorepostaction,pattern=horizontal lines},}}}
%\tikzset{xxxxUgalxxxAlexxlabelxxxxx/.style={automatically generated plot,green,dashed,mark=square}}
%\tikzset{xxxxUgalxxxAlexxlabelxxxxxbar/.style={automatically generated bar plot,fill=green!20,postaction={pattern=grid},}}
%\tikzset{xxxxUgalxxxAlexxlabelxxxxxboxplot/.style={automatically generated boxplot,fill=green!20,every path/.style={postaction={nomorepostaction,pattern=grid},}}}
%\tikzset{xxxxUgalxxxValiantxxxxx/.style={automatically generated plot,blue,dotted,mark=triangle}}
%\tikzset{xxxxUgalxxxValiantxxxxxbar/.style={automatically generated bar plot,fill=blue!20,postaction={pattern=crosshatch},}}
%\tikzset{xxxxUgalxxxValiantxxxxxboxplot/.style={automatically generated boxplot,fill=blue!20,every path/.style={postaction={nomorepostaction,pattern=crosshatch},}}}
%
%%\tikzset{xxxxValiantxxxxxbar/.style={automatically generated bar plot,fill=black!20,postaction={pattern=crosshatch},}}
%\tikzset{xxxxValiantxxxxxbar/.style={automatically generated bar plot,fill=black!20,postaction={pattern=north east lines},}}

\def\xShiftxRIxtxt{}%"Shift-1"
\def\xShiftxRIx{0}
\def\xShiftxRIIxtxt{"Shift-2"}
\def\xShiftxRIIx{1}
\def\xShiftxRIIIxtxt{"Shift-3"}
\def\xShiftxRIIIx{2}
\def\xShiftxRIVxtxt{"Shift-4"}
\def\xShiftxRIVx{3}
\def\xSwitchxcomplementxxRIxtxt{"Switch complement +1"}
\def\xSwitchxcomplementxxRIx{4}
\def\xSwitchxcomplementxtxt{"Switch complement"}
\def\xSwitchxcomplementx{5}
\def\xSwitchxcomplementxxRIaxtxt{} %"Switch complement -1a"
\def\xSwitchxcomplementxxRIax{6}
\def\xRandomxswitchxpermutationxtxt{} %"Random switch permutation"
\def\xRandomxswitchxpermutationx{7}

\def\xxxxValiantxxxxxtext{Valiant}

%% -- henceafter the data

\begin{experimentfigure}%
	%\begin{center}
	\centering%
	\tikzpicturedependsonfile{externalized-plots/external-RRIIIxburstxfacil-avgthroughput-selectorxxxxShiftxRIxxxxx-kind0.md5}%
	\tikzsetnextfilename{externalized-legends/legend-RRIIIxburstxfacil-avgthroughput-xxxxShiftxRIxxxxx}%
	\pgfplotslegendfromname{legend-RRIIIxburstxfacil-avgthroughput-xxxxShiftxRIxxxxx}\\
	\tikzsetnextfilename{externalized-plots/external-RRIIIxburstxfacil-avgthroughput-selectorxxxxShiftxRIxxxxx-kind0}%
	\begin{tikzpicture}[baseline,remember picture]
	\begin{axis}[
		automatically generated symbolic,xtick={0,...,0}, xticklabels = {{\xShiftxRIxtxt}},ybar,bar width=5pt,enlarge x limits=0.2,
		first kind,,
		legend to name=legend-RRIIIxburstxfacil-avgthroughput-xxxxShiftxRIxxxxx,
		title={Shift},
		%%ybar interval=0.6,
		ymin=0,%ymax=400000,
		ymajorgrids=true,
		yminorgrids=true,
		xmajorgrids=true,
		mark options=solid,
		minor y tick num=4,
		ylabel={Cycles to finish},
		%ytick={0e4,2e4,4e4,6e4},
		%%legend style={at={(1.05,1.0)},anchor=north west},
		%%legend style={opacity=0.7,at={(0.99,0.99)},anchor=north east},
		%%legend style={at={(0.00,1.01)},anchor=south west,font=\scriptsize},legend columns=3,transpose legend,legend cell align=left,
		%%legend style={at={(0.00,1.01)},anchor=south west,font=\scriptsize},legend columns=2,legend cell align=left,
		%%every x tick label/.append style={anchor=base,yshift=-7},
	]
\addplot[xxxxBrinrxrandomxxxxxbar] coordinates{(\xShiftxRIx,436672.5)};\addlegendentry{\xxxxBrinrxrandomxxxxxtext}
%\addplot[xxxxBrinrxrandomxxxxxbar] coordinates{(\xShiftxRIx,85000)};\addlegendentry{\xxxxBrinrxrandomxxxxxtext}%%% FAKE VALUES for an axix break.
\addplot[xxxxUgalxxxAlexxlabelxxxxxbar] coordinates{(\xShiftxRIx,47937.168)};\addlegendentry{\xxxxUgalxxxAlexxlabelxxxxxtext}
\addplot[xxxxUgalxxxValiantxxxxxbar] coordinates{(\xShiftxRIx,61176.168)};\addlegendentry{\xxxxUgalxxxValiantxxxxxtext}
\addplot[xxxxValiantxxxxxbar] coordinates{(\xShiftxRIx,46634.5)};\addlegendentry{\xxxxValiantxxxxxtext}
	%\draw[thick, double=white] (-.5,7e4) -- ++(.25,.5e4) node[anchor=south west,xshift=-3pt] {\tiny $44\cdot 10^4$};
	\end{axis}
	%\draw[thick, white, fill=white] (4.5,5) -- (4.6,5.2) -- (4.7,5) -- (4.8,5.2) -- (4.9,5);
	%\draw[thick, double=white] (0,0) -- (5,5);
%	\begin{axis}[
%		height=105pt,
%		width=85pt,
%		yticklabel style={font=\tiny,/pgf/number format/.cd, fixed,/tikz/.cd},% formattin ticks' labels
%		y label style={at={(ticklabel cs:0.5, +8pt)},name={y label},anchor=north,font=\scriptsize, rotate=180},
%		axis y line*=right,
%		axis x line=none,
%		ylabel={Throughput},
%		ymin=0, ymax=40000,
%		ytick=		  {0, 		 6553.6, 13107.2, 26214.4}, % Ensure ticks are within the range
%		yticklabels = {$\infty$, 	0.5, 	0.25, 	 0.12},
%%		ytick=		  {0, 		 3276.8, 6553.6, 13107.2, 26214.4, 40000 }, % Ensure ticks are within the range
%%		yticklabels = {$\infty$, 1.0, 	 0.5, 	 0.25, 	  0.12,   0.08   },
%%		yticklabel = { %\pgfmathprintnumber{\tick}:
%%			\pgfmathtruncatemacro\iszero{\tick==0}%
%%			\ifthenelse{\iszero=1}{1.0}{%
%%				\pgfmathparse{3276.8/\tick}\pgfmathprintnumber{\pgfmathresult}%
%%			}%
%%		},
%		scaled y ticks=false,
%	]
%		\addplot[draw=none] {x};   % dummy plot
%	\end{axis}
	%\pgfresetboundingbox\useasboundingbox (y label.north west) (current axis.north east) ($(current axis.outer north west)!(current axis.north east)!(current axis.outer north east)$);
	\end{tikzpicture}%[ 	"Shift-1", ]  - 0
	\tikzpicturedependsonfile{externalized-plots/external-RRIIIxburstxfacil-avgthroughput-selectorxxxxSwitchxcomplementxxRIaxxxxx-kind0.md5}%
	\tikzsetnextfilename{externalized-legends/legend-RRIIIxburstxfacil-avgthroughput-xxxxSwitchxcomplementxxRIaxxxxx}%
%	\pgfplotslegendfromname{legend-RRIIIxburstxfacil-avgthroughput-xxxxSwitchxcomplementxxRIaxxxxx}\\
	\tikzsetnextfilename{externalized-plots/external-RRIIIxburstxfacil-avgthroughput-selectorxxxxSwitchxcomplementxxRIaxxxxx-kind0}%
	\begin{tikzpicture}[baseline,remember picture]
	\begin{axis}[
		automatically generated symbolic,xtick={6,...,6}, xticklabels = {{\xSwitchxcomplementxxRIaxtxt}},ybar,bar width=5pt,enlarge x limits=0.2,
		first kind,,
		legend to name=legend-RRIIIxburstxfacil-avgthroughput-xxxxSwitchxcomplementxxRIaxxxxx,
		title={Complement},
		%%ybar interval=0.6,
		ymin=0,%ymax=180000,
		ymajorgrids=true,
		yminorgrids=true,
		xmajorgrids=true,
		mark options=solid,
		minor y tick num=4,
		ylabel={Cycles to finish},
		%%legend style={at={(1.05,1.0)},anchor=north west},
		%%legend style={opacity=0.7,at={(0.99,0.99)},anchor=north east},
		%%legend style={at={(0.00,1.01)},anchor=south west,font=\scriptsize},legend columns=3,transpose legend,legend cell align=left,
		%%legend style={at={(0.00,1.01)},anchor=south west,font=\scriptsize},legend columns=2,legend cell align=left,
		%%every x tick label/.append style={anchor=base,yshift=-7},
	]
	\addplot[xxxxBrinrxrandomxxxxxbar] coordinates{(\xRandomxswitchxpermutationx,111660.836)};\addlegendentry{\xxxxBrinrxrandomxxxxxtext}
	\addplot[xxxxUgalxxxAlexxlabelxxxxxbar] coordinates{(\xRandomxswitchxpermutationx,111609.164)};\addlegendentry{\xxxxUgalxxxAlexxlabelxxxxxtext}
	\addplot[xxxxUgalxxxValiantxxxxxbar] coordinates{(\xRandomxswitchxpermutationx,60551.832)};\addlegendentry{\xxxxUgalxxxValiantxxxxxtext}
	\addplot[xxxxValiantxxxxxbar] coordinates{(\xRandomxswitchxpermutationx,46923.832)};\addlegendentry{\xxxxValiantxxxxxtext}
	\end{axis}
%	\begin{axis}[
%		height=105pt,
%		width=85pt,
%		yticklabel style={font=\tiny,/pgf/number format/.cd, fixed,/tikz/.cd},% formattin ticks' labels
%		y label style={at={(ticklabel cs:0.5, +8pt)},name={y label},anchor=north,font=\scriptsize, rotate=180},
%		axis y line*=right,
%		axis x line=none,
%		ylabel={Throughput},
%		ymin=0, ymax=18000,
%		ytick=		  {0, 		 6553.6, 13107.2}, % Ensure ticks are within the range
%		yticklabels = {$\infty$, 	0.5, 	0.25},
%%		ytick=		  {0, 		 3276.8, 6553.6, 13107.2, 26214.4, 40000 }, % Ensure ticks are within the range
%%		yticklabels = {$\infty$, 1.0, 	 0.5, 	 0.25, 	  0.12,   0.08   },
%%		yticklabel = { %\pgfmathprintnumber{\tick}:
%%			\pgfmathtruncatemacro\iszero{\tick==0}%
%%			\ifthenelse{\iszero=1}{1.0}{%
%%				\pgfmathparse{3276.8/\tick}\pgfmathprintnumber{\pgfmathresult}%
%%			}%
%%		},
%		scaled y ticks=false,
%	]
%		\addplot[draw=none] {x};   % dummy plot
%	\end{axis}
	%\pgfresetboundingbox\useasboundingbox (y label.north west) (current axis.north east) ($(current axis.outer north west)!(current axis.north east)!(current axis.outer north east)$);
	\end{tikzpicture}%[ 	"Switch complement -1a", ]  - 0
	\tikzpicturedependsonfile{externalized-plots/external-RRIIIxburstxfacil-avgthroughput-selectorxxxxRandomxswitchxpermutationxxxxx-kind0.md5}%
	\tikzsetnextfilename{externalized-legends/legend-RRIIIxburstxfacil-avgthroughput-xxxxRandomxswitchxpermutationxxxxx}%
%	\pgfplotslegendfromname{legend-RRIIIxburstxfacil-avgthroughput-xxxxRandomxswitchxpermutationxxxxx}\\
	\tikzsetnextfilename{externalized-plots/external-RRIIIxburstxfacil-avgthroughput-selectorxxxxRandomxswitchxpermutationxxxxx-kind0}%
	\begin{tikzpicture}[baseline,remember picture]
	\begin{axis}[
		automatically generated symbolic,xtick={7,...,7}, xticklabels = {{\xRandomxswitchxpermutationxtxt}},ybar,bar width=5pt,enlarge x limits=0.2,
		first kind,,
		legend to name=legend-RRIIIxburstxfacil-avgthroughput-xxxxRandomxswitchxpermutationxxxxx,
		title={RSP},
		%%ybar interval=0.6,
		ymin=0, %ymax=180000,
		ymajorgrids=true,
		yminorgrids=true,
		xmajorgrids=true,
		mark options=solid,
		minor y tick num=4,
		ylabel={Cycles to finish},
		%%legend style={at={(1.05,1.0)},anchor=north west},
		%%legend style={opacity=0.7,at={(0.99,0.99)},anchor=north east},
		%%legend style={at={(0.00,1.01)},anchor=south west,font=\scriptsize},legend columns=3,transpose legend,legend cell align=left,
		%%legend style={at={(0.00,1.01)},anchor=south west,font=\scriptsize},legend columns=2,legend cell align=left,
		%%every x tick label/.append style={anchor=base,yshift=-7},
	]
\addplot[xxxxBrinrxrandomxxxxxbar] coordinates{(\xRandomxswitchxpermutationx,340390.5)};\addlegendentry{\xxxxBrinrxrandomxxxxxtext}
\addplot[xxxxUgalxxxAlexxlabelxxxxxbar] coordinates{(\xRandomxswitchxpermutationx,88205.164)};\addlegendentry{\xxxxUgalxxxAlexxlabelxxxxxtext}
\addplot[xxxxUgalxxxValiantxxxxxbar] coordinates{(\xRandomxswitchxpermutationx,60551.832)};\addlegendentry{\xxxxUgalxxxValiantxxxxxtext}
\addplot[xxxxValiantxxxxxbar] coordinates{(\xRandomxswitchxpermutationx,47265.5)};\addlegendentry{\xxxxValiantxxxxxtext}
	\end{axis}
%	\begin{axis}[
%		height=105pt,
%		width=85pt,
%		yticklabel style={font=\tiny,/pgf/number format/.cd, fixed,/tikz/.cd},% formattin ticks' labels
%		y label style={at={(ticklabel cs:0.5, +8pt)},name={y label},anchor=north,font=\scriptsize, rotate=180},
%		axis y line*=right,
%		axis x line=none,
%		ylabel={Throughput},
%		ymin=0, ymax=18000,
%		ytick=		  {0, 		 6553.6, 13107.2}, % Ensure ticks are within the range
%		yticklabels = {$\infty$, 	0.5, 	0.25},
%%		ytick=		  {0, 		 3276.8, 6553.6, 13107.2, 26214.4, 40000 }, % Ensure ticks are within the range
%%		yticklabels = {$\infty$, 1.0, 	 0.5, 	 0.25, 	  0.12,   0.08   },
%%		yticklabel = { %\pgfmathprintnumber{\tick}:
%%			\pgfmathtruncatemacro\iszero{\tick==0}%
%%			\ifthenelse{\iszero=1}{1.0}{%
%%				\pgfmathparse{3276.8/\tick}\pgfmathprintnumber{\pgfmathresult}%
%%			}%
%%		},
%		scaled y ticks=false,
%	]
%		\addplot[draw=none] {x};   % dummy plot
%	\end{axis}
	%\pgfresetboundingbox\useasboundingbox (y label.north west) (current axis.north east) ($(current axis.outer north west)!(current axis.north east)!(current axis.outer north east)$);
	\end{tikzpicture}%[ 	"Random switch permutation", ]  - 0
	%\end{center}
	\caption{Time to finish of different traffic patterns in a FM$_{64}$ with 4096 servers.}
	\label{fig:ttf_simple_apps}%
\end{experimentfigure}

\egroup

    In Figure~\ref{fig:ttf_simple_apps} three different synthetic traffic patterns are simulated: shift, complement, and random permutation under fixed generation.
    The results show the number of cycles to finish the same amount of load for the different routings.

    %%Carmen: MIRAR LA CAPTION, PONE THOUGHTPUT

    As it can be seen, sRINR always provides smaller or same completion times than bRINR, standing out 9 times faster in the shift pattern.
    In the case of the RSP pattern, sRINR is 3.8 times faster.
    Complement traffic pattern reveals as the more challenging situation, in which both sRINR and bRINR take more than 2.3 times to complete than Valiant.

    Since the traffic patterns considered are adversarial, Valiant naturally provides the best completion times, although doubling  the resources of ordering schemes.
    It can be shown that sRINR under shift traffic patterns achieves best-case throughput of 0.5~flits/cycle/server, which matches Valiant performance.
    However, its worst-case performance, under complement traffic, drops to 0.25~flits/cycle/server, halving Valiant throughput.
    Thus, neither sRINR nor bRINR can be considered competitive.
    For comparison purposes, only sRINR will be considered in later evaluations.

    %The shift pattern constitutes the best case for sRINR, being 9 times faster than bRINR.
%    In addition, the complement traffic is the worst case for sRINR, but the performance is the same as bRINR.
%    The random permutation represents the average case, in which sRINR is 3.8 times faster.
%    Although sRINR is a bit slower than bRINR in its worst case, it is faster in average and much faster in the best case.
%    Nevertheless, as we will see, none of them is competitive against TERA.
%    Thus, only sRINR ordering will be considered as a baseline competitor of TERA in subsequent evaluations.

    \subsection{Service topology selection}\label{subsec:seleccion}

    Figure~\ref{fig:performance_safe_evaluation} shows the performance of TERA with different service topologies.
    It includes two plots where the size of the FM is increased.
    RSP and FR traffic patterns under fixed generation are employed.

    For RSP, the performance of the 2-Tree (Path or 1D-Mesh) topology is the highest, and the 2D-HyperX is the lowest.
    However, the performance of 2D-HyperX reduces the gap with the path topology as the size of the FM increases, and the 3D-HyperX follows a similar trend.

    In the FR scenario, the path and the 4-Tree have the worst performance due to their asymmetry.
    Other admissible scenarios could have been chosen to stress the network, but the FR pattern represents the most challenging case.
    The Up-Down service paths condition the network performance.
    Similar results have been obtained in $d$-dimensional meshes.

    \bgroup

%% %% -- common pgfplots prelude --
%% %\newenvironment{experimentfigure}{\begin{figure}[H]\tikzexternalenable}{\tikzexternaldisable\end{figure}}
%% %\newenvironment{experimentfigure}{\begin{figure*}}{\end{figure*}}
%% %\newcommand{\kindseparator}{\hskip 0ex{}}
%% \pgfplotsset{compat=newest}
%% \makeatletter
%% %required for fill plus pattern on boxplot
%% \tikzset{nomorepostaction/.code=\let\tikz@postactions\pgfutil@empty}
%% \makeatother
%% \pgfplotsset{minor grid style={dashed,very thin, color=blue!15}}
%% \pgfplotsset{major grid style={very thin, color=black!30}}
\pgfplotsset{
	automatically generated axis/.style={
		%default: height=207pt, width=240pt. 240:207 ~~ 7:6
		%height=115pt,%may fit 3figures with 1 line caption
%		height=105pt,%may fit 3figures with 2 line caption
%		width=190pt, %174pt,
		height=90pt,%105pt,%may fit 3figures with 2 line caption
		width=130pt,%174pt,
		scaled ticks=false,
		xticklabel style={font=\tiny,/pgf/number format/.cd, fixed,/tikz/.cd},% formattin ticks' labels
		yticklabel style={font=\tiny,/pgf/number format/.cd, fixed,/tikz/.cd},% formattin ticks' labels
		x label style={at={(ticklabel cs:0.5, -5pt)},name={x label},anchor=north,font=\scriptsize},
		y label style={at={(ticklabel cs:0.5, -5pt)},name={y label},anchor=south,font=\scriptsize},
	},
	automatically generated symbolic/.style={
		height=105pt,
		width=100pt,
		xticklabel style={font=\tiny,rotate=90},
		yticklabel style={font=\tiny,/pgf/number format/.cd, fixed,/tikz/.cd},% formattin ticks' labels
		x label style={at={(ticklabel cs:0.5, -5pt)},name={x label},anchor=north,font=\scriptsize},
		y label style={at={(ticklabel cs:0.5, -5pt)},name={y label},anchor=south,font=\scriptsize},
	},
	first kind/.style={
		%The first axis on each line of plots
		%legend style={overlay,at={(0.50,1.05)},anchor=south,font=\scriptsize,fill=none},
		%legend style={at={(0.00,1.01)},anchor=south west,font=\scriptsize,fill=none},
		%legend style={at={($(axis description cs:0.00,1.01)!(current page.center)!(axis description cs:1.00,1.01)$)},anchor=south,font=\scriptsize,fill=none},
		legend style={font=\scriptsize,fill=none},
		legend columns=3,legend cell align=left,
	},
	posterior kind/.style={
		%Axis following the first on each line of plots
		%legend style={at={(0.50,1.05)},overlay,anchor=south,font=\tiny,fill=none},
		legend style={draw=none},
	},
}
\newcommand\captionprologue{X: }
\newcommand\experimenttitle{00-pruebas-k8/bar_throughput.pdf (all 18 done)}
\newcommand\experimentheader{\tiny 00-pruebas-k8:bar\_throughput.pdf (all 18 done)\\pdflatex on \today\\version=heads/alex-stable-dirty-ffe590dee2973b77f19bc89796b6a76f715e09ab(0.6.3)}
%\tikzset{xMeshxopportunisticx/.style={automatically generated plot,red,solid,mark=o}}
%\tikzset{xMeshxopportunisticxbar/.style={automatically generated bar plot,fill=red!20,postaction={pattern=horizontal lines},}}
%\tikzset{xMeshxopportunisticxboxplot/.style={automatically generated boxplot,fill=red!20,every path/.style={postaction={nomorepostaction,pattern=horizontal lines},}}}
%\tikzset{xEmbeddedxHyperXxRIIDxsourcex/.style={automatically generated plot,green,dashed,mark=square}}
%\tikzset{xEmbeddedxHyperXxRIIDxsourcexbar/.style={automatically generated bar plot,fill=green!20,postaction={pattern=grid},}}
%\tikzset{xEmbeddedxHyperXxRIIDxsourcexboxplot/.style={automatically generated boxplot,fill=green!20,every path/.style={postaction={nomorepostaction,pattern=grid},}}}
%\tikzset{xValiantx/.style={automatically generated plot,blue,dotted,mark=triangle}}
%\tikzset{xValiantxbar/.style={automatically generated bar plot,fill=blue!20,postaction={pattern=crosshatch},}}
%\tikzset{xValiantxboxplot/.style={automatically generated boxplot,fill=blue!20,every path/.style={postaction={nomorepostaction,pattern=crosshatch},}}}
%
%\tikzset{xUgalValiantx/.style={automatically generated plot,black,dotted,mark=triangle, mark options={solid}}}
%\tikzset{xRIVxtreexopportunisticx/.style={automatically generated plot,orange,dotted,mark=triangle, mark options={solid}}}
%\tikzset{hamingthreedimensions/.style={automatically generated plot,black,dotted,mark=square, mark options={solid}}}

\def\xFixedxRandomxtxt{}
\def\xFixedxRandomx{0}
\def\xRandomxswitchxpermutationxtxt{"Random switch permutation"}
\def\xRandomxswitchxpermutationx{0}
\def\xShiftxRIxtxt{}
\def\xShiftxRIx{0}
\def\xkochox{0}
\def\xkdieciseisx{1}
\def\xktreintaydosx{2}
\def\xksesentaycuatrox{3}
\def\xValiantxtext{Valiant}
%\def\xMeshxopportunisticxtext{1D-Mesh}

%% -- henceafter the data

\begin{experimentfigure}%
	%\begin{center}
	\centering%
	\tikzpicturedependsonfile{externalized-plots/external-RRxpruebasxkRVIII-barthroughput-selectorxxxx-kind0.md5}%
	\tikzsetnextfilename{externalized-legends/legend-RRxpruebasxkRVIII-barthroughput-xxxx}%
	\pgfplotslegendfromname{legend-RRxpruebasxkRVIII-barthroughput-xxxx}\\
	\tikzsetnextfilename{externalized-plots/random-switchperm}%
	\begin{tikzpicture}[baseline,remember picture]
	\begin{axis}[
		automatically generated axis, xtick={0, 1, 2, 3}, xticklabels = {8, 16, 32, 64},
		first kind,,
		legend to name=legend-RRxpruebasxkRVIII-barthroughput-xxxx,
		title={RSP},
		%%ybar interval=0.6,
		ymin=0,%
		ymajorgrids=true,
		yminorgrids=true,
		xmajorgrids=true,
		mark options=solid,
		minor y tick num=4,
		xlabel={FM size},
		ylabel={Cycle to finish},
		%divide by 4 the y axis
		yticklabel={\pgfmathparse{\tick/4}\pgfmathprintnumber{\pgfmathresult}},
		scaled y ticks=true,
%		y tick scale label code/.code={\tiny$\times10^3$},
	]
\addplot[xMeshxopportunisticx] coordinates{  (\xkochox,167711.33) (\xkdieciseisx,172371.33) (\xktreintaydosx,179194) (\xksesentaycuatrox, 174855.33)};\addlegendentry{\xMeshxopportunisticxtext}
\addplot[xEmbeddedxHyperXxRIIDxsourcex] coordinates{ (\xkochox,229604.67) (\xkdieciseisx,220618) (\xktreintaydosx,210266) (\xksesentaycuatrox, 195323.33)};\addlegendentry{\xEmbeddedxHyperXxRIIDxsourcextext}
\addplot[hamingthreedimensions] coordinates{ (\xkochox, 205095) (\xkdieciseisx,202727.33) (\xktreintaydosx, 191239.33) (\xksesentaycuatrox, 184687.33)};\addlegendentry{\TERAHXIIItext}
\addplot[xRIVxtreexopportunisticx] coordinates{(\xkochox,203690) (\xkdieciseisx, 186832.67) (\xktreintaydosx,181195.33) (\xksesentaycuatrox, 183532.67)};\addlegendentry{\TERATREEtext}
\addplot[xValiantx] coordinates{ (\xkochox,176444.67) (\xkdieciseisx,178986) (\xktreintaydosx,184570) (\xksesentaycuatrox, 189062)};\addlegendentry{\xValiantxtext}
\addplot[xUgalValiantx] coordinates{ (\xkochox,185820.67) (\xkdieciseisx,199928.67) (\xktreintaydosx,217632.67) (\xksesentaycuatrox, 240554)};\addlegendentry{\UGALtext}
	\end{axis}
	%\pgfresetboundingbox\useasboundingbox (y label.north west) (current axis.north east) ($(current axis.outer north west)!(current axis.north east)!(current axis.outer north east)$);
	\end{tikzpicture}%[ ]  - 0
%	Fixed Random
	\tikzsetnextfilename{externalized-plots/external-RRbxpruebasxkRVIII-barthroughput-selectorxxxx-kind0}%
	\begin{tikzpicture}[baseline,remember picture]
	\begin{axis}[
		automatically generated axis, xtick={0, 1, 2, 3}, xticklabels = {8, 16, 32, 64},
		first kind,,
		legend to name=legend-RRbxpruebasxkRVIII-barthroughput-xxxx,
		title={FR},
		%%ybar interval=0.6,
		ymin=0,%
		ymajorgrids=true,
		yminorgrids=true,
		xmajorgrids=true,
		mark options=solid,
		minor y tick num=4,
		xlabel={FM size},
		ylabel={Cycle to finish},
		yticklabel={\pgfmathparse{\tick/4}\pgfmathprintnumber{\pgfmathresult}},
		scaled y ticks=true,
		%%legend style={at={(1.05,1.0)},anchor=north west},
		%%legend style={opacity=0.7,at={(0.99,0.99)},anchor=north east},
		%%legend style={at={(0.00,1.01)},anchor=south west,font=\scriptsize},legend columns=3,transpose legend,legend cell align=left,
		%%legend style={at={(0.00,1.01)},anchor=south west,font=\scriptsize},legend columns=2,legend cell align=left,
		%%every x tick label/.append style={anchor=base,yshift=-7},
	]
\addplot[xMeshxopportunisticx] coordinates{(\xkochox,384810) (\xkdieciseisx,3099631.3) (\xktreintaydosx,5688168.5) (\xksesentaycuatrox, 19212064)};\addlegendentry{\xMeshxopportunisticxtext}
\addplot[xEmbeddedxHyperXxRIIDxsourcex] coordinates{(\xkochox,329671.34) (\xkdieciseisx,448472.66) (\xktreintaydosx,442286) (\xksesentaycuatrox, 538230)};\addlegendentry{\xEmbeddedxHyperXxRIIDxsourcextext}
\addplot[hamingthreedimensions] coordinates{ (\xkochox, 345231) (\xkdieciseisx, 452434) (\xktreintaydosx, 438844.66) (\xksesentaycuatrox, 533719.3)};\addlegendentry{\TERAHXIIItext}
\addplot[xRIVxtreexopportunisticx] coordinates{(\xkochox,360492.66) (\xkdieciseisx, 913824.7) (\xktreintaydosx,3258278) (\xksesentaycuatrox,16188512)};\addlegendentry{\TERATREEtext}
\addplot[xValiantx] coordinates{(\xkochox,342408.66) (\xkdieciseisx,460132.66) (\xktreintaydosx,491116.66) (\xksesentaycuatrox, 593762)};\addlegendentry{\xValiantxtext}
\addplot[xUgalValiantx] coordinates{(\xkochox,340356.66) (\xkdieciseisx,444074) (\xktreintaydosx,441183.34) (\xksesentaycuatrox, 533359.3)};\addlegendentry{\xUgalxxxValiantxtext}
	\end{axis}
	%\pgfresetboundingbox\useasboundingbox (y label.north west) (current axis.north east) ($(current axis.outer north west)!(current axis.north east)!(current axis.outer north east)$);
	\end{tikzpicture}%[ ]  - 0
	\caption{Time to consume a burst of packets under Random Switch Permutation and Fixed Random patterns for different FM sizes.}%
	\label{fig:performance_safe_evaluation}%
\end{experimentfigure}

\egroup

%    \input{../simulations/topologies-comparison/hop_histogram_k32_fixed}

%     Among the topology families considered, the HyperX family emerges as the most suitable candidate for constructing the safe graph in TERA.
%    Despite being the most restrictive in terms of non-minimal throughput due to its high number of links, HyperX topologies enable TERA to achieve higher non-minimal path diversity than any other link ordering scheme, as they are limited by Claims~\ref{claim:maxpaths} and~\ref{claim:intermediates_forward}.
%
%    In addition to providing diverse path options, HyperX also ensures a fair distribution of paths across all switches, contributing to balanced network utilization.
%    While TERA does not fulfill the conditions of Claim~\ref{claim:utilizationfairpaths}, it still guarantees that the main links are evenly used among switches, maintaining efficiency.
%
%    Another key advantage of using HyperX is its symmetry and low diameter, which help constrain the number of hops required on safe paths.
%%    Although the use of safe paths imposes some limitations, TERA always prefers to use minimal hops whenever they are not saturated.
%    Allowing additional hops does not necessarily lead to higher latency, as packets are only diverted when minimal paths are congested, enabling them to traverse less occupied buffers.
%    This behavior ensures that longer paths are only used when beneficial, minimizing queuing delays and maintaining overall performance under load.
%    This is going to be analyzed in Section~\ref{sec:results}.

    In conclusion, despite the trade-offs in non-minimal throughput due to the denser connectivity of HyperX, its  properties—symmetry, low diameter, and balanced link usage—make it the best choice for the service topology in TERA.
    2D-HyperX and 3D-HyperX are selected for the remainder of the experiments.
    Compared to link ordering schemes, TERA consistently offers better path diversity, resilience to congestion, and fairer utilization of network resources.

    \subsection{Bernoulli traffic evaluation}\label{sec:static}

    \input{common_adaptive_hammings_k64.tex}
    Figure~\ref{fig:static_traffic_patterns} shows the performance of the different routing mechanisms for the UN and RSP traffic patterns under Bernoulli generation.

    With UN, all routing algorithms use minimal paths for approximately 80–90\% of the packets, resulting in similar performance.
    Minimal routing serves as the baseline for evaluating algorithms that utilize only 1 VC.
    Omni-WAR and UGAL, which use 2 VCs, slightly outperform the others by reducing some head-of-line (HoL) blocking on non-minimal paths.
    All mechanisms exhibit fair behavior with a comparable hop distribution at maximum injection load.
    The TERA variants show a negligible use of longer paths (3- and 4-hop routes occur in less than 1\% of cases), which has no impact on performance and does not appear under low to medium traffic loads.

    For the RSP pattern, routing performance varies more significantly.
    %Valiant uses 2 VCs, and its performance is the baseline for evaluating the other algorithms.
    The best-performing algorithms are, in order, Omni-WAR, TERA-HX3, Valiant, TERA-HX2, UGAL, and sRINR.
    Omni-WAR and TERA-HX3 slightly outperform Valiant because they can adaptively select the best non-minimal path at the source switch.
%    The link ordering mechanism performs poorly due to its limited use of non-minimal paths compared to other approaches.
    TERA outperforms sRINR by an 80\%, due to the limitations of sRINR.
%    Regarding fairness, UGAL experiences a ~6\% drop in the Jain’s fairness index, due to the fact that [...].
    Again, TERA routing shows extremely low occurrences of 3 and 4-hop paths (less than 0.1\%), which are inconsequential to performance.

%    In Figure~\ref{fig:link_used_safe} the utilization of safe and main links for the S-HyperX-2D routing is shown.
    Lastly, the utilization of the service and main links for TERA-HX routing algorithm was measured.
    Under uniform traffic, all links had similar utilization since TERA selects minimal paths without distinguishing between main and service links.
    Under RSP traffic, service links show about half the utilization of main links.
    Nevertheless, since service links constitute less than 10\% of the total (192 out of 2016 in TERA-HX3), their limited usage has a negligible impact on overall link utilization and performance.

    \subsection{Application results}\label{subsec:apps}

    \bgroup

%% %% -- common pgfplots prelude --
%% %\newenvironment{experimentfigure}{\begin{figure}[H]\tikzexternalenable}{\tikzexternaldisable\end{figure}}
%% %\newenvironment{experimentfigure}{\begin{figure*}}{\end{figure*}}
%% %\newcommand{\kindseparator}{\hskip 0ex{}}
%% \pgfplotsset{compat=newest}
%% \makeatletter
%% %required for fill plus pattern on boxplot
%% \tikzset{nomorepostaction/.code=\let\tikz@postactions\pgfutil@empty}
%% \makeatother
%% \pgfplotsset{minor grid style={dashed,very thin, color=blue!15}}
%% \pgfplotsset{major grid style={very thin, color=black!30}}
\pgfplotsset{
%% 	automatically generated axis/.style={
%% 		%default: height=207pt, width=240pt. 240:207 ~~ 7:6
%% 		%height=115pt,%may fit 3figures with 1 line caption
%% 		height=105pt,%may fit 3figures with 2 line caption
%% 		width=174pt,
%% 		scaled ticks=false,
%% 		xticklabel style={font=\tiny,/pgf/number format/.cd, fixed,/tikz/.cd},% formattin ticks' labels
%% 		yticklabel style={font=\tiny,/pgf/number format/.cd, fixed,/tikz/.cd},% formattin ticks' labels
%% 		x label style={at={(ticklabel cs:0.5, -5pt)},name={x label},anchor=north,font=\scriptsize},
%% 		y label style={at={(ticklabel cs:0.5, -5pt)},name={y label},anchor=south,font=\scriptsize},
%% 	},
 	automatically generated symbolic/.style={
 		height=110pt,%120pt,%105pt,
 		width=120pt,
 		xticklabel style={font=\tiny}, % ,rotate=90
 		yticklabel style={font=\tiny,/pgf/number format/.cd, fixed,/tikz/.cd},% formattin ticks' labels
 		x label style={at={(ticklabel cs:0.5, -5pt)},name={x label},anchor=north,font=\scriptsize},
 		y label style={at={(ticklabel cs:0.5, -5pt)},name={y label},anchor=south,font=\scriptsize},
 	},
%% 	first kind/.style={
%% 		%The first axis on each line of plots
%% 		%legend style={overlay,at={(0.50,1.05)},anchor=south,font=\scriptsize,fill=none},
%% 		%legend style={at={(0.00,1.01)},anchor=south west,font=\scriptsize,fill=none},
%% 		%legend style={at={($(axis description cs:0.00,1.01)!(current page.center)!(axis description cs:1.00,1.01)$)},anchor=south,font=\scriptsize,fill=none},
%% 		legend style={font=\scriptsize,fill=none},
%% 		legend columns=6,legend cell align=left,
%% 	},
%% 	posterior kind/.style={
%% 		%Axis following the first on each line of plots
%% 		%legend style={at={(0.50,1.05)},overlay,anchor=south,font=\tiny,fill=none},
%% 		legend style={draw=none},
%% 	},
}
\newcommand\captionprologue{X: }
\newcommand\experimenttitle{03a-fusion-todo/cycles_to_finish.pdf (all 180 done)}
\newcommand\experimentheader{\tiny 03a-fusion-todo:cycles\_to\_finish.pdf (all 180 done)\\pdflatex on \today\\version=heads/alex-stable-dirty-ffe590dee2973b77f19bc89796b6a76f715e09ab(0.6.3)}

\def\xNaturalxtxt{Linear}
\def\xNaturalx{0}
\def\xRandomxtxt{Random}
\def\xRandomx{1}
%\def\xxxxBrinrxrandomxxxxxtext{[
%	"Brinr random",
%]
%}
%\def\xxxxUgalxxxAlexxlabelxxxxxtext{[
%	"Ugal - Alex label",
%]
%}
%\def\xxxxEmbeddedxHyperXxRIIDxxxxxtext{[
%	"Embedded HyperX 2D",
%]
%}
%\def\xxxxEmbeddedxHyperXxRIIIDxxxxxtext{[
%	"Embedded HyperX 3D",
%	]
%}
%
%\def\xxxxUgalxxxDistxlabelxxxxxtext{[
%	"Ugal - Dist label",
%]
%}
%\def\xxxxUgalxxxValiantxxxxxtext{[
%	"Ugal - Valiant",
%]
%}
%\def\xxxxomnixxxxxtext{[
%	"omni",
%]
%}
% \def\xxxxBrinrxrandomxxxxxtext{bRINR}
% \def\xxxxUgalxxxAlexxlabelxxxxxtext{sRINR}
% \def\xxxxomnixxxxxtext{Omni-WAR}
% \def\xxxxValiantxxxxxtext{MIN/VLB}
% \def\xxxxEmbeddedxHyperXxRIIDxxxxxtext{S-HyperX-2D}
% \def\xxxxEmbeddedxHyperXxRIIIDxxxxxtext{S-HyperX-3D}
% \def\xxxxUgalxxxValiantxxxxxtext{UGAL}

%% -- henceafter the data

\begin{experimentfigurecomp}%
	%\begin{center}
	\centering%

	\tikzpicturedependsonfile{externalized-plots/external-RRIIIaxfusionxtodo-aaaa-selectorxxxxAllRIIAllxRIRVIxxxxx-kind0.md5}%
	\tikzsetnextfilename{externalized-legends/legend-RRIIIaxfusionxtodo-aaaa-xxxxAllRIIAllxRIRVIxxxxx}%
	\pgfplotslegendfromname{legend-RRIIIaxfusionxtodo-aaaa-xxxxAllRIIAllxRIRVIxxxxx}\\
	\tikzsetnextfilename{externalized-plots/external-RRIIIaxfusionxtodo-aaaa-selectorxxxxAllRIIAllxRIRVIxxxxx-kind0}%
	\begin{tikzpicture}[baseline,remember picture]
	\begin{axis}[
		automatically generated symbolic,xtick={0,...,1}, xticklabels = {{\xNaturalxtxt},{\xRandomxtxt}},ybar,bar width=3.5pt,enlarge x limits=0.5,
		first kind,,
		legend columns=7,
		legend to name=legend-RRIIIaxfusionxtodo-aaaa-xxxxAllRIIAllxRIRVIxxxxx,
		title={All2All},
		%%ybar interval=0.6,
		ymin=0,%
		ymajorgrids=true,
		yminorgrids=true,
		xmajorgrids=true,
		mark options=solid,
		minor y tick num=4,
		xlabel={},
		ylabel={Cycles to finish},
		%%legend style={at={(1.05,1.0)},anchor=north west},
		%%legend style={opacity=0.7,at={(0.99,0.99)},anchor=north east},
		%%legend style={at={(0.00,1.01)},anchor=south west,font=\scriptsize},legend columns=3,transpose legend,legend cell align=left,
		%%legend style={at={(0.00,1.01)},anchor=south west,font=\scriptsize},legend columns=2,legend cell align=left,
		%%every x tick label/.append style={anchor=base,yshift=-7},
	]
\addplot[xxxxomnixxxxxbar] coordinates{(\xNaturalx,254339) (\xRandomx,161611.67)};\addlegendentry{\xxxxomnixxxxxtext}
\addplot[xxxxUgalxxxValiantxxxxxbar] coordinates{(\xNaturalx,289055.66) (\xRandomx,155365.67)};\addlegendentry{\xxxxUgalxxxValiantxxxxxtext}
%\addplot[xxxxBrinrxrandomxxxxxbar] coordinates{(\xNaturalx,409909) (\xRandomx,338921)};\addlegendentry{\xxxxBrinrxrandomxxxxxtext}
\addplot[xxxxUgalxxxAlexxlabelxxxxxbar] coordinates{(\xNaturalx,298104.34) (\xRandomx,169717.67)};\addlegendentry{\xxxxUgalxxxAlexxlabelxxxxxtext}
\addplot[xxxxEmbeddedxHyperXxRIIDxxxxxbar] coordinates{(\xNaturalx,273797.66) (\xRandomx,164425)};\addlegendentry{\xxxxEmbeddedxHyperXxRIIDxxxxxtext}
\addplot[xxxxEmbeddedxHyperXxRIIIDxxxxxbar] coordinates{(\xNaturalx,267509) (\xRandomx,164549.67)};\addlegendentry{\xxxxEmbeddedxHyperXxRIIIDxxxxxtext}
%\addplot[xxxxUgalxxxDistxlabelxxxxxbar] coordinates{(\xNaturalx,300096.34) (\xRandomx,169947.67)};\addlegendentry{\xxxxUgalxxxDistxlabelxxxxxtext}
	\end{axis}
	%\pgfresetboundingbox\useasboundingbox (y label.north west) (current axis.north east) ($(current axis.outer north west)!(current axis.north east)!(current axis.outer north east)$);
	\end{tikzpicture}%[ 	"All2All 16", ]  - 0
	\tikzpicturedependsonfile{externalized-plots/external-RRIIIaxfusionxtodo-aaaa-selectorxxxxStencilxkingxRIIDxxxxx-kind0.md5}%
	\tikzsetnextfilename{externalized-legends/legend-RRIIIaxfusionxtodo-aaaa-xxxxStencilxkingxRIIDxxxxx}%
%	\pgfplotslegendfromname{legend-RRIIIaxfusionxtodo-aaaa-xxxxStencilxkingxRIIDxxxxx}\\
	\tikzsetnextfilename{externalized-plots/external-RRIIIaxfusionxtodo-aaaa-selectorxxxxStencilxkingxRIIDxxxxx-kind0}%
	\begin{tikzpicture}[baseline,remember picture]
	\begin{axis}[
		automatically generated symbolic,xtick={0,...,1}, xticklabels = {{\xNaturalxtxt},{\xRandomxtxt}},ybar,bar width=3.5pt,enlarge x limits=0.5,
		first kind,,
		legend to name=legend-RRIIIaxfusionxtodo-aaaa-xxxxStencilxkingxRIIDxxxxx,
		title={Stencil 2D},
		%%ybar interval=0.6,
		ymin=0,%
		ymajorgrids=true,
		yminorgrids=true,
		xmajorgrids=true,
		mark options=solid,
		minor y tick num=4,
		xlabel={},
		ylabel={Cycles to finish},
		%%legend style={at={(1.05,1.0)},anchor=north west},
		%%legend style={opacity=0.7,at={(0.99,0.99)},anchor=north east},
		%%legend style={at={(0.00,1.01)},anchor=south west,font=\scriptsize},legend columns=3,transpose legend,legend cell align=left,
		%%legend style={at={(0.00,1.01)},anchor=south west,font=\scriptsize},legend columns=2,legend cell align=left,
		%%every x tick label/.append style={anchor=base,yshift=-7},
	]
\addplot[xxxxomnixxxxxbar] coordinates{(\xNaturalx,35212.332) (\xRandomx,32253)};\addlegendentry{\xxxxomnixxxxxtext}
\addplot[xxxxUgalxxxValiantxxxxxbar] coordinates{(\xNaturalx,49503) (\xRandomx,40590.332)};\addlegendentry{\xxxxUgalxxxValiantxxxxxtext}
%\addplot[xxxxBrinrxrandomxxxxxbar] coordinates{(\xNaturalx,201859.67) (\xRandomx,84949.664)};\addlegendentry{\xxxxBrinrxrandomxxxxxtext}
\addplot[xxxxUgalxxxAlexxlabelxxxxxbar] coordinates{(\xNaturalx,54544.332) (\xRandomx,48254.332)};\addlegendentry{\xxxxUgalxxxAlexxlabelxxxxxtext}
\addplot[xxxxEmbeddedxHyperXxRIIDxxxxxbar] coordinates{(\xNaturalx,41701) (\xRandomx,36736.332)};\addlegendentry{\xxxxEmbeddedxHyperXxRIIDxxxxxtext}
\addplot[xxxxEmbeddedxHyperXxRIIIDxxxxxbar] coordinates{(\xNaturalx,40101) (\xRandomx,36935.668)};\addlegendentry{\xxxxEmbeddedxHyperXxRIIIDxxxxxtext}
%\addplot[xxxxUgalxxxDistxlabelxxxxxbar] coordinates{(\xNaturalx,54544.332) (\xRandomx,48999.668)};\addlegendentry{\xxxxUgalxxxDistxlabelxxxxxtext}
	\end{axis}
	%\pgfresetboundingbox\useasboundingbox (y label.north west) (current axis.north east) ($(current axis.outer north west)!(current axis.north east)!(current axis.outer north east)$);
	\end{tikzpicture}%[ 	"Stencil king 2D", ]  - 0
	\tikzpicturedependsonfile{externalized-plots/external-RRIIIaxfusionxtodo-aaaa-selectorxxxxStencilxkingxRIIIDxxxxx-kind0.md5}%
	\tikzsetnextfilename{externalized-legends/legend-RRIIIaxfusionxtodo-aaaa-xxxxStencilxkingxRIIIDxxxxx}%
%	\pgfplotslegendfromname{legend-RRIIIaxfusionxtodo-aaaa-xxxxStencilxkingxRIIIDxxxxx}\\
	\tikzsetnextfilename{externalized-plots/external-RRIIIaxfusionxtodo-aaaa-selectorxxxxStencilxkingxRIIIDxxxxx-kind0}%
	\begin{tikzpicture}[baseline,remember picture]
	\begin{axis}[
		automatically generated symbolic,xtick={0,...,1}, xticklabels = {{\xNaturalxtxt},{\xRandomxtxt}},ybar,bar width=3.5pt,enlarge x limits=0.5,
		first kind,,
		legend to name=legend-RRIIIaxfusionxtodo-aaaa-xxxxStencilxkingxRIIIDxxxxx,
		title={Stencil 3D},
		%%ybar interval=0.6,
		ymin=0,%
		ymajorgrids=true,
		yminorgrids=true,
		xmajorgrids=true,
		mark options=solid,
		minor y tick num=4,
		xlabel={},
		ylabel={Cycles to finish},
		%%legend style={at={(1.05,1.0)},anchor=north west},
		%%legend style={opacity=0.7,at={(0.99,0.99)},anchor=north east},
		%%legend style={at={(0.00,1.01)},anchor=south west,font=\scriptsize},legend columns=3,transpose legend,legend cell align=left,
		%%legend style={at={(0.00,1.01)},anchor=south west,font=\scriptsize},legend columns=2,legend cell align=left,
		%%every x tick label/.append style={anchor=base,yshift=-7},
	]
\addplot[xxxxomnixxxxxbar] coordinates{(\xNaturalx,149439) (\xRandomx,124001.664)};\addlegendentry{\xxxxomnixxxxxtext}
\addplot[xxxxUgalxxxValiantxxxxxbar] coordinates{(\xNaturalx,172725.67) (\xRandomx,143417)};\addlegendentry{\xxxxUgalxxxValiantxxxxxtext}
%\addplot[xxxxBrinrxrandomxxxxxbar] coordinates{(\xNaturalx,420328.34) (\xRandomx,312806.34)};\addlegendentry{\xxxxBrinrxrandomxxxxxtext}
\addplot[xxxxUgalxxxAlexxlabelxxxxxbar] coordinates{(\xNaturalx,217508.33) (\xRandomx,166607.67)};\addlegendentry{\xxxxUgalxxxAlexxlabelxxxxxtext}
\addplot[xxxxEmbeddedxHyperXxRIIDxxxxxbar] coordinates{(\xNaturalx,169967) (\xRandomx,142616.33)};\addlegendentry{\xxxxEmbeddedxHyperXxRIIDxxxxxtext}
\addplot[xxxxEmbeddedxHyperXxRIIIDxxxxxbar] coordinates{(\xNaturalx,161435) (\xRandomx,140655.67)};\addlegendentry{\xxxxEmbeddedxHyperXxRIIIDxxxxxtext}
%\addplot[xxxxUgalxxxDistxlabelxxxxxbar] coordinates{(\xNaturalx,220565.67) (\xRandomx,168309.67)};\addlegendentry{\xxxxUgalxxxDistxlabelxxxxxtext}
	\end{axis}
	%\pgfresetboundingbox\useasboundingbox (y label.north west) (current axis.north east) ($(current axis.outer north west)!(current axis.north east)!(current axis.outer north east)$);
	\end{tikzpicture}%[ 	"Stencil king 3D", ]  - 0
	\tikzpicturedependsonfile{externalized-plots/external-RRIIIaxfusionxtodo-aaaa-selectorxxxxFFTRIIIDxxxxx-kind0.md5}%
	\tikzsetnextfilename{externalized-legends/legend-RRIIIaxfusionxtodo-aaaa-xxxxFFTRIIIDxxxxx}%
%	\pgfplotslegendfromname{legend-RRIIIaxfusionxtodo-aaaa-xxxxFFTRIIIDxxxxx}\\
	\tikzsetnextfilename{externalized-plots/external-RRIIIaxfusionxtodo-aaaa-selectorxxxxFFTRIIIDxxxxx-kind0}%
	\begin{tikzpicture}[baseline,remember picture]
	\begin{axis}[
		automatically generated symbolic,xtick={0,...,1}, xticklabels = {{\xNaturalxtxt},{\xRandomxtxt}},ybar,bar width=3.5pt,enlarge x limits=0.5,
		first kind,,
		legend to name=legend-RRIIIaxfusionxtodo-aaaa-xxxxFFTRIIIDxxxxx,
		title={FFT3D},
		%%ybar interval=0.6,
		ymin=0,%
		ymajorgrids=true,
		yminorgrids=true,
		xmajorgrids=true,
		mark options=solid,
		minor y tick num=4,
		xlabel={},
		ylabel={Cycles to finish},
		%%legend style={at={(1.05,1.0)},anchor=north west},
		%%legend style={opacity=0.7,at={(0.99,0.99)},anchor=north east},
		%%legend style={at={(0.00,1.01)},anchor=south west,font=\scriptsize},legend columns=3,transpose legend,legend cell align=left,
		%%legend style={at={(0.00,1.01)},anchor=south west,font=\scriptsize},legend columns=2,legend cell align=left,
		%%every x tick label/.append style={anchor=base,yshift=-7},
	]
\addplot[xxxxomnixxxxxbar] coordinates{(\xNaturalx,71095) (\xRandomx,83739.664)};\addlegendentry{\xxxxomnixxxxxtext}
\addplot[xxxxUgalxxxValiantxxxxxbar] coordinates{(\xNaturalx,70223) (\xRandomx,80938.336)};\addlegendentry{\xxxxUgalxxxValiantxxxxxtext}
%\addplot[xxxxBrinrxrandomxxxxxbar] coordinates{(\xNaturalx,128316.336) (\xRandomx,162769)};\addlegendentry{\xxxxBrinrxrandomxxxxxtext}
\addplot[xxxxUgalxxxAlexxlabelxxxxxbar] coordinates{(\xNaturalx,75670.336) (\xRandomx,87362.336)};\addlegendentry{\xxxxUgalxxxAlexxlabelxxxxxtext}
\addplot[xxxxEmbeddedxHyperXxRIIDxxxxxbar] coordinates{(\xNaturalx,75043.664) (\xRandomx,83665.664)};\addlegendentry{\xxxxEmbeddedxHyperXxRIIDxxxxxtext}
\addplot[xxxxEmbeddedxHyperXxRIIIDxxxxxbar] coordinates{(\xNaturalx,74879) (\xRandomx,83223.664)};\addlegendentry{\xxxxEmbeddedxHyperXxRIIIDxxxxxtext}
%\addplot[xxxxUgalxxxDistxlabelxxxxxbar] coordinates{(\xNaturalx,75563) (\xRandomx,88633)};\addlegendentry{\xxxxUgalxxxDistxlabelxxxxxtext}
%\addplot[xxxxomnixxxxxbar] coordinates{(\xNaturalx,103498.336) (\xRandomx,103207)};\addlegendentry{\xxxxomnixxxxxtext}
%\addplot[xxxxUgalxxxValiantxxxxxbar] coordinates{(\xNaturalx,109792.336) (\xRandomx,124924.336)};\addlegendentry{\xxxxUgalxxxValiantxxxxxtext}
%%\addplot[xxxxBrinrxrandomxxxxxbar] coordinates{(\xNaturalx,191795.67) (\xRandomx,297660.34)};\addlegendentry{\xxxxBrinrxrandomxxxxxtext}
%\addplot[xxxxUgalxxxAlexxlabelxxxxxbar] coordinates{(\xNaturalx,130649) (\xRandomx,175485)};\addlegendentry{\xxxxUgalxxxAlexxlabelxxxxxtext}
%\addplot[xxxxEmbeddedxHyperXxRIIDxxxxxbar] coordinates{(\xNaturalx,120007) (\xRandomx,129773.664)};\addlegendentry{\xxxxEmbeddedxHyperXxRIIDxxxxxtext}
%\addplot[xxxxEmbeddedxHyperXxRIIIDxxxxxbar] coordinates{(\xNaturalx,118749.664) (\xRandomx,128467.664)};\addlegendentry{\xxxxEmbeddedxHyperXxRIIIDxxxxxtext}
%\addplot[xxxxUgalxxxDistxlabelxxxxxbar] coordinates{(\xNaturalx,131041.664) (\xRandomx,174091)};\addlegendentry{\xxxxUgalxxxDistxlabelxxxxxtext}
	\end{axis}
	%\pgfresetboundingbox\useasboundingbox (y label.north west) (current axis.north east) ($(current axis.outer north west)!(current axis.north east)!(current axis.outer north east)$);
	\end{tikzpicture}%[ 	"FFT3D", ]  - 0
	\tikzpicturedependsonfile{externalized-plots/external-RRIIIaxfusionxtodo-aaaa-selectorxxxxAllreducexxxxx-kind0.md5}%
	\tikzsetnextfilename{externalized-legends/legend-RRIIIaxfusionxtodo-aaaa-xxxxAllreducexxxxx}%
%	\pgfplotslegendfromname{legend-RRIIIaxfusionxtodo-aaaa-xxxxAllreducexxxxx}\\
	\tikzsetnextfilename{externalized-plots/external-RRIIIaxfusionxtodo-aaaa-selectorxxxxAllreducexxxxx-kind0}%
	\begin{tikzpicture}[baseline,remember picture]
	\begin{axis}[
		automatically generated symbolic,xtick={0,...,1}, xticklabels = {{\xNaturalxtxt},{\xRandomxtxt}},ybar,bar width=3.5pt,enlarge x limits=0.5,
		first kind,,
		legend to name=legend-RRIIIaxfusionxtodo-aaaa-xxxxAllreducexxxxx,
		title={Allreduce},
		%%ybar interval=0.6,
		ymin=0,%
		ymajorgrids=true,
		yminorgrids=true,
		xmajorgrids=true,
		mark options=solid,
		minor y tick num=4,
		xlabel={},
		ylabel={Cycles to finish},
		%%legend style={at={(1.05,1.0)},anchor=north west},
		%%legend style={opacity=0.7,at={(0.99,0.99)},anchor=north east},
		%%legend style={at={(0.00,1.01)},anchor=south west,font=\scriptsize},legend columns=3,transpose legend,legend cell align=left,
		%%legend style={at={(0.00,1.01)},anchor=south west,font=\scriptsize},legend columns=2,legend cell align=left,
		%%every x tick label/.append style={anchor=base,yshift=-7},
	]
\addplot[xxxxomnixxxxxbar] coordinates{(\xNaturalx,407265.66) (\xRandomx,543099)};\addlegendentry{\xxxxomnixxxxxtext}
\addplot[xxxxUgalxxxValiantxxxxxbar] coordinates{(\xNaturalx,627199) (\xRandomx,752035)};\addlegendentry{\xxxxUgalxxxValiantxxxxxtext}
%\addplot[xxxxBrinrxrandomxxxxxbar] coordinates{(\xNaturalx,1157301.6) (\xRandomx,1039467)};\addlegendentry{\xxxxBrinrxrandomxxxxxtext}
\addplot[xxxxUgalxxxAlexxlabelxxxxxbar] coordinates{(\xNaturalx,813860.3) (\xRandomx,797597.7)};\addlegendentry{\xxxxUgalxxxAlexxlabelxxxxxtext}
\addplot[xxxxEmbeddedxHyperXxRIIDxxxxxbar] coordinates{(\xNaturalx,438237.66) (\xRandomx,585297.7)};\addlegendentry{\xxxxEmbeddedxHyperXxRIIDxxxxxtext}
\addplot[xxxxEmbeddedxHyperXxRIIIDxxxxxbar] coordinates{(\xNaturalx,423987.66) (\xRandomx,582464.3)};\addlegendentry{\xxxxEmbeddedxHyperXxRIIIDxxxxxtext}
%\addplot[xxxxUgalxxxDistxlabelxxxxxbar] coordinates{(\xNaturalx,859283.7) (\xRandomx,800880.3)};\addlegendentry{\xxxxUgalxxxDistxlabelxxxxxtext}
	\end{axis}
	%\pgfresetboundingbox\useasboundingbox (y label.north west) (current axis.north east) ($(current axis.outer north west)!(current axis.north east)!(current axis.outer north east)$);
	\end{tikzpicture}%[ 	"Allreduce", ]  - 0
	%\end{center}
	\caption{Cycles to consume an application kernel in a FM$_{64}$}%
	\label{fig:apps_k64_all}%
\end{experimentfigurecomp}

\egroup

    \input{violin_summary_output_2.tex}

Figure~\ref{fig:apps_k64_all} shows the application completion times across different routing algorithms.
Despite not using virtual channels, TERA achieves competitive performance, closely following Omni-WAR and outperforming UGAL in most scenarios.
This proves the effectiveness of TERA as a low-cost routing mechanism for Full-mesh networks.

Omni-WAR delivers the best overall performance due to the use of 2 VCs and unrestricted non-minimal bandwidth.
It consistently outperforms all other mechanisms, particularly in the Stencil 2D and 3D workloads, where its advantage increases to around 10\%.

In spite of not using VCs, TERA (HX2 and HX3) trails Omni-WAR by a small margin, within 7\% on average, and outperforms UGAL significantly, with speedups of up to 47\% in the Allreduce application.
This improvement stems from TERA’s ability to select an optimal non-minimal path from a large pool of candidates at the source switch, effectively reducing packet latency and improving load balancing.
This capability is particularly beneficial in heavy communication workloads like Allreduce, where message dependencies can cause added delays.

In TERA, packets taking 3 or 4 hops constitute less than 1\% of the total on average.
To assess the impact of these longer paths on packet latency across all applications, violin plots are presented in Figure~\ref{fig:percentileslatency}.
As no significant performance differences were observed between linear and random process mappings, evaluations focus on the linear mapping.
The link ordering algorithm is omitted due to its lack of competitiveness.

%Figure~\ref{fig:percentileslatency} presents violin plots of packet latencies.
TERA-HX2 and TERA-HX3 exhibit the lowest mean and 99\% percentile latencies in most cases.
This can be attributed to their reduced buffer space, which would result in lower queuing delays.
At the higher percentiles (99.9\% and 99.99\%), TERA remains the top performer except in the Stencil 3D workload, where its latency is comparable to Omni-WAR.
UGAL consistently shows the highest latency across all cases, primarily due to its reliance on a single randomly selected intermediate in Valiant-style routing, limiting its ability to adapt to network congestion.

%Finally, Figure~\ref{fig:hist_hops} provides hop count distributions for the evaluated applications.
%All routings predominantly use paths of lengths 1 and 2, as expected, due to their adaptive use of minimal and non-minimal routes.
%There are also paths of length 0 in some traffics, which represent local communications inside the same switch.
%TERA occasionally utilizes paths larger than 2 hops, which appear infrequently (around $1\%$ or less).
%%These longer paths have a negligible impact on performance.
TERA offers competitive performance, making it a practical and efficient solution for Full-mesh topologies.
Hence, TERA reveals as a low-cost alternative to existing adaptive algorithms like UGAL and Omni-WAR.

\subsection{Other topologies}
\label{sec:discussion}

%\subsection{Fault tolerance}
%As far as there is a service topology, the network can provide of fault tolerance.
%If a link of the main graph fails, the network can still route packets and avoid deadlocks through the service topology.
%If a link of the service topology fails, the service topology needs to be recalculated.
%
%
%
%\subsection{Application to other topologies}

Any topology based on FM networks, such as HyperX, Dragonfly and Dragonfly+, can take advantage of TERA.
A simulation of an $8\times8$ 2D-HyperX network with 512 servers is shown in Figure~\ref{fig:label_hamming_2d_perf}.

Dim-WAR, a routing algorithm specially conceived for the 2D-HyperX that uses 2 VCs~\cite{Kim_omni}, is also included in the comparison.
Additionally, we have developed a O1TURN~\cite{O1turn} version of TERA that uses 2 VCs.
Thus, Omni-WAR uses 4 VCs, Dim-WAR and O1TURN-TERA-HX3 use 2 VCs, and DOR-TERA-HX3 just 1VC.

In DOR-TERA-HX3, the TERA-HX3 routing algorithm is applied independently within each of the two FM$_8$ subnetworks traversed by a packet—one for each dimension—with dimensions visited in $XY$ order.
In the case of O1TURN-TERA-HX3, at the source switch, it is decided if the packet will progress either in XY or YX order, and a VC is needed for each dimension order.

It is noticeable that DOR-TERA-HX3 reveals competitive in most cases being the one with minimal resources.
Additionally, the performance of O1TURN-TERA-HX3 is near Omni-WAR with half the resources (from 4 VCs to 2 VCs), and up to 32\% better than Dim-WAR with the same buffering.

\bgroup

%% -- common pgfplots prelude --
%\newenvironment{experimentfigure}{\begin{figure}[H]\tikzexternalenable}{\tikzexternaldisable\end{figure}}
%\newenvironment{experimentfigure}{\begin{figure*}}{\end{figure*}}
%\newcommand{\kindseparator}{\hskip 0ex{}}
\pgfplotsset{compat=newest}
\makeatletter
%required for fill plus pattern on boxplot
\tikzset{nomorepostaction/.code=\let\tikz@postactions\pgfutil@empty}
\makeatother
\pgfplotsset{minor grid style={dashed,very thin, color=blue!15}}
\pgfplotsset{major grid style={very thin, color=black!30}}
\pgfplotsset{
	automatically generated axis/.style={
		%default: height=207pt, width=240pt. 240:207 ~~ 7:6
		%height=115pt,%may fit 3figures with 1 line caption
		height=105pt,%may fit 3figures with 2 line caption
		width=174pt,
		scaled ticks=false,
		xticklabel style={font=\tiny,/pgf/number format/.cd, fixed,/tikz/.cd},% formattin ticks' labels
		yticklabel style={font=\tiny,/pgf/number format/.cd, fixed,/tikz/.cd},% formattin ticks' labels
		x label style={at={(ticklabel cs:0.5, -5pt)},name={x label},anchor=north,font=\scriptsize},
		y label style={at={(ticklabel cs:0.5, -5pt)},name={y label},anchor=south,font=\scriptsize},
	},
	automatically generated symbolic/.style={
		height=90pt,%120pt,%105pt,
		width=120pt,
		xticklabel style={font=\tiny}, % ,rotate=90
		yticklabel style={font=\tiny,/pgf/number format/.cd, fixed,/tikz/.cd},% formattin ticks' labels
		x label style={at={(ticklabel cs:0.5, -5pt)},name={x label},anchor=north,font=\scriptsize},
		y label style={at={(ticklabel cs:0.5, -5pt)},name={y label},anchor=south,font=\scriptsize},
	},
	first kind/.style={
		%The first axis on each line of plots
		%legend style={overlay,at={(0.50,1.05)},anchor=south,font=\scriptsize,fill=none},
		%legend style={at={(0.00,1.01)},anchor=south west,font=\scriptsize,fill=none},
		%legend style={at={($(axis description cs:0.00,1.01)!(current page.center)!(axis description cs:1.00,1.01)$)},anchor=south,font=\scriptsize,fill=none},
		legend style={font=\scriptsize,fill=none},
		legend columns=2,legend cell align=left,
	},
	posterior kind/.style={
		%Axis following the first on each line of plots
		%legend style={at={(0.50,1.05)},overlay,anchor=south,font=\tiny,fill=none},
		legend style={draw=none},
	},
}
\def\timetickcode{%
	%\pgfkeys{/pgf/fpu}%
	\pgfkeys{/pgf/fpu,/pgf/fpu/output format=fixed}%
	%\pgfkeys{/pgf/fpu=true}%
	\pgfmathparse{\tick}%
	%\typeout{tick=\tick, math=\pgfmathresult}%
	\edef\tmp{\pgfmathresult}%
	%\pgfmathsetmacro\xseconds{Mod(\tmp,60)}%
	%\pgfkeys{/pgf/fpu=false}%
	%\pgfmathtruncatemacro\xseconds{Mod(\tmp,60)}%
	%\pgfkeys{/pgf/fpu=true}%
	%\typeout{total seconds=\tmp}%
	%\pgfmathparse{\tmp/60}\typeout{tmp/60=\pgfmathresult}%
	%\pgfmathparse{floor(\tmp/60)}\typeout{floor(tmp/60)=\pgfmathresult}%
	%\pgfmathparse{round(\tmp/60)}\typeout{round(tmp/60)=\pgfmathresult}%
	%\pgfmathparse{int(\tmp/60)}\typeout{int(tmp/60)=\pgfmathresult}%
	\pgfmathtruncatemacro\seconds{\tmp-60*floor(\tmp/60)}%
	%\typeout{truncated seconds=\seconds}%
	\pgfmathtruncatemacro\tmp{(\tmp - \seconds)/60}%
	%\typeout{total minutes=\tmp}%
	%\pgfmathtruncatemacro\minutes{Mod(\tmp,60)}%
	\pgfmathtruncatemacro\minutes{\tmp-60*floor(\tmp/60)}%
	%\typeout{truncated minutes=\seconds}%
	\pgfmathtruncatemacro\tmp{(\tmp - \minutes)/60}%
	%\typeout{total hours=\tmp}%
	%\pgfmathtruncatemacro\hours{Mod(\tmp,24)}%
	\pgfmathtruncatemacro\hours{\tmp-24*floor(\tmp/24)}%
	%\typeout{truncated hours=\hours}%
	\pgfmathtruncatemacro\days{(\tmp - \hours)/24}%
	%{\tiny\days-\hours:\minutes:\seconds}%
	\ifnum\days=0%
		{\tiny\hours:\minutes:\seconds}%
	\else%
		{\tiny\days-\hours:\minutes}%
	\fi%
	%{\tiny\tick}%
}
\def\memorytickcode{%
	\pgfkeys{/pgf/fpu,/pgf/fpu/output format=fixed}%
	%\pgfmathprintnumber[sci,sci zerofill,precision=1]{\tick}%
	\pgfmathtruncatemacro\unitcase{log2(\tick+0.001)/10+1.1}%We add a byte to keep the log2 sensible.
	%\pgfmathfloatifflags{}
	\pgfmathparse{\tick / pow(1024,\unitcase-1)}%
	\pgfmathprintnumber{\pgfmathresult}%
	\ifcase\unitcase B%0
		\or KB%1
		\or MB%2
		\or GB%3
		\or TB%4
		\else +1024,\pgfmathprintnumber{unitcase}B%
	\fi%
}
\tikzset{
	automatically generated plot/.style={
		%/pgfplots/error bars/.cd,error bar style={ultra thick},x dir=both, y dir=both,
		/pgfplots/error bars/x dir=both,
		/pgfplots/error bars/y dir=both,
		/pgfplots/error bars/x explicit,
		/pgfplots/error bars/y explicit,
		/pgfplots/error bars/error bar style={ultra thin,solid},
		/tikz/mark options={solid},
	},
	automatically generated bar plot/.style={
		/pgfplots/error bars/y dir=both,
		/pgfplots/error bars/y explicit,
	},
	automatically generated boxplot/.style={
		%/pgfplots/boxplot={
		%	box extend=0.1,
		%	every average/.style={/tikz/mark=*},
		%},
	},
	%/pgf/images/aux in dpth=true,
	x time ticks/.style={
		/pgfplots/scaled x ticks=false,
		/pgfplots/xticklabel={\timetickcode},
	},
	y time ticks/.style={
		/pgfplots/scaled y ticks=false,
		/pgfplots/yticklabel={\timetickcode}
	},
	x memory ticks from kilobytes/.style={
		/pgfplots/scaled x ticks=false,
		/pgfplots/xticklabel={\memorytickcode}
	},
	y memory ticks from kilobytes/.style={
		/pgfplots/scaled y ticks=false,
		/pgfplots/yticklabel={\memorytickcode}
	},
}

%% -- experiment-local prelude
\newcommand\captionprologue{X: }
\newcommand\experimenttitle{09-pruebas-apps-8x8/cycles_to_finish.pdf (all 120 done)}
\newcommand\experimentheader{\tiny 09-pruebas-apps-8x8:cycles\_to\_finish.pdf (all 120 done)\\pdflatex on \today\\version=heads/alex-stable-dirty-d6bc24e0eefa4de5234e465f429ade2c2fa0184b(0.6.3)}
%\tikzset{xxxxOmnixWARxxxxx/.style={automatically generated plot,red,solid,mark=o}}
%\tikzset{xxxxOmnixWARxxxxxbar/.style={automatically generated bar plot,fill=red!20,postaction={pattern=horizontal lines},}}
%\tikzset{xxxxOmnixWARxxxxxboxplot/.style={automatically generated boxplot,fill=red!20,every path/.style={postaction={nomorepostaction,pattern=horizontal lines},}}}
%\tikzset{xxxxDimWARxxxxx/.style={automatically generated plot,green,dashed,mark=square}}
%\tikzset{xxxxDimWARxxxxxbar/.style={automatically generated bar plot,fill=green!20,postaction={pattern=grid},}}
%\tikzset{xxxxDimWARxxxxxboxplot/.style={automatically generated boxplot,fill=green!20,every path/.style={postaction={nomorepostaction,pattern=grid},}}}
%\tikzset{xxxxORIxturnxSxHypercubexxxxx/.style={automatically generated plot,blue,dotted,mark=triangle}}
%\tikzset{xxxxORIxturnxSxHypercubexxxxxbar/.style={automatically generated bar plot,fill=blue!20,postaction={pattern=crosshatch},}}
%\tikzset{xxxxORIxturnxSxHypercubexxxxxboxplot/.style={automatically generated boxplot,fill=blue!20,every path/.style={postaction={nomorepostaction,pattern=crosshatch},}}}
%\tikzset{xxxxDORxSxHypercubexxxxx/.style={automatically generated plot,black,dash dot,mark=star}}
%\tikzset{xxxxDORxSxHypercubexxxxxbar/.style={automatically generated bar plot,fill=black!20,postaction={pattern=dots},}}
%\tikzset{xxxxDORxSxHypercubexxxxxboxplot/.style={automatically generated boxplot,fill=black!20,every path/.style={postaction={nomorepostaction,pattern=dots},}}}

%\def\xNaturalxtxt{Linear}
\def\xNaturalx{0}
\def\xRandomx{1}
%\def\xxxxDORxSxHypercubexxxxxtext{DOR-S-Hypercube}
%\def\xxxxOmnixWARxxxxxtext{Omni-WAR}
%\def\xxxxDimWARxxxxxtext{DimWAR}
%\def\xxxxORIxturnxSxHypercubexxxxxtext{O1-turn-S-Hypercube}

%% -- henceafter the data

\begin{experimentfigure}%
	%\begin{center}
	\centering%
	\tikzpicturedependsonfile{externalized-plots/external-RRIXxpruebasxappsxRVIIIxRVIII-aaaa-selectorxxxxAllRIIAllxRIRVIxxxxx-kind0.md5}%
	\tikzsetnextfilename{externalized-legends/legend-RRIXxpruebasxappsxRVIIIxRVIII-aaaa-xxxxAllRIIAllxRIRVIxxxxx}%
	\pgfplotslegendfromname{legend-RRIXxpruebasxappsxRVIIIxRVIII-aaaa-xxxxAllRIIAllxRIRVIxxxxx}\\
	\tikzsetnextfilename{externalized-plots/external-RRIXxpruebasxappsxRVIIIxRVIII-aaaa-selectorxxxxAllRIIAllxRIRVIxxxxx-kind0}%
	\begin{tikzpicture}[baseline,remember picture]
	\begin{axis}[
		automatically generated symbolic,xtick={0,...,1}, xticklabels = {{\xNaturalxtxt},{\xRandomxtxt}},ybar,bar width=3.5pt,enlarge x limits=0.5,
		first kind,,
		legend to name=legend-RRIXxpruebasxappsxRVIIIxRVIII-aaaa-xxxxAllRIIAllxRIRVIxxxxx,
		title={All2All},
		%%ybar interval=0.6,
		ymin=0,%
		ymajorgrids=true,
		yminorgrids=true,
		xmajorgrids=true,
		mark options=solid,
		minor y tick num=4,
		xlabel={},
		ylabel={Cycles to finish},
		%%legend style={at={(1.05,1.0)},anchor=north west},
		%%legend style={opacity=0.7,at={(0.99,0.99)},anchor=north east},
		%%legend style={at={(0.00,1.01)},anchor=south west,font=\scriptsize},legend columns=3,transpose legend,legend cell align=left,
		%%legend style={at={(0.00,1.01)},anchor=south west,font=\scriptsize},legend columns=2,legend cell align=left,
		%%every x tick label/.append style={anchor=base,yshift=-7},
	]
\addplot[xxxxOmnixWARxxxxxbar] coordinates{(\xNaturalx,270989) (\xRandomx,181611)};\addlegendentry{\xxxxOmnixWARxxxxxtext}
\addplot[xxxxDimWARxxxxxbar] coordinates{(\xNaturalx,345701.66) (\xRandomx,223302.33)};\addlegendentry{\xxxxDimWARxxxxxtext}
\addplot[xxxxORIxturnxSxHypercubexxxxxbar] coordinates{(\xNaturalx,279864.34) (\xRandomx,203101)};\addlegendentry{\xxxxORIxturnxSxHypercubexxxxxtext}
\addplot[xxxxDORxSxHypercubexxxxxbar] coordinates{(\xNaturalx,360319.66) (\xRandomx,267841.66)};\addlegendentry{\xxxxDORxSxHypercubexxxxxtext}
	\end{axis}
	%\pgfresetboundingbox\useasboundingbox (y label.north west) (current axis.north east) ($(current axis.outer north west)!(current axis.north east)!(current axis.outer north east)$);
	\end{tikzpicture}%[ 	"All2All 16", ]  - 0
		\tikzpicturedependsonfile{externalized-plots/external-RRIXxpruebasxappsxRVIIIxRVIII-aaaa-selectorxxxxAllreducexxxxx-kind0.md5}%
	\tikzsetnextfilename{externalized-legends/legend-RRIXxpruebasxappsxRVIIIxRVIII-aaaa-xxxxAllreducexxxxx}%
%	\pgfplotslegendfromname{legend-RRIXxpruebasxappsxRVIIIxRVIII-aaaa-xxxxAllreducexxxxx}\\
	\tikzsetnextfilename{externalized-plots/external-RRIXxpruebasxappsxRVIIIxRVIII-aaaa-selectorxxxxAllreducexxxxx-kind0}%
	\begin{tikzpicture}[baseline,remember picture]
	\begin{axis}[
		automatically generated symbolic,xtick={0,...,1}, xticklabels = {{\xNaturalxtxt},{\xRandomxtxt}},ybar,bar width=3.5pt,enlarge x limits=0.5,
		first kind,,
		legend to name=legend-RRIXxpruebasxappsxRVIIIxRVIII-aaaa-xxxxAllreducexxxxx,
		title={Allreduce},
		%%ybar interval=0.6,
		ymin=0,%
		ymajorgrids=true,
		yminorgrids=true,
		xmajorgrids=true,
		mark options=solid,
		minor y tick num=4,
		xlabel={},
		ylabel={Cycles to finish},
		%%legend style={at={(1.05,1.0)},anchor=north west},
		%%legend style={opacity=0.7,at={(0.99,0.99)},anchor=north east},
		%%legend style={at={(0.00,1.01)},anchor=south west,font=\scriptsize},legend columns=3,transpose legend,legend cell align=left,
		%%legend style={at={(0.00,1.01)},anchor=south west,font=\scriptsize},legend columns=2,legend cell align=left,
		%%every x tick label/.append style={anchor=base,yshift=-7},
	]
\addplot[xxxxOmnixWARxxxxxbar] coordinates{(\xNaturalx,107578.336) (\xRandomx,132090.33)};\addlegendentry{\xxxxOmnixWARxxxxxtext}
\addplot[xxxxDimWARxxxxxbar] coordinates{(\xNaturalx,105391) (\xRandomx,192593.67)};\addlegendentry{\xxxxDimWARxxxxxtext}
\addplot[xxxxORIxturnxSxHypercubexxxxxbar] coordinates{(\xNaturalx,113281.664) (\xRandomx,145592.33)};\addlegendentry{\xxxxORIxturnxSxHypercubexxxxxtext}
\addplot[xxxxDORxSxHypercubexxxxxbar] coordinates{(\xNaturalx,111513.664) (\xRandomx,223140.33)};\addlegendentry{\xxxxDORxSxHypercubexxxxxtext}
	\end{axis}
	%\pgfresetboundingbox\useasboundingbox (y label.north west) (current axis.north east) ($(current axis.outer north west)!(current axis.north east)!(current axis.outer north east)$);
	\end{tikzpicture}%[ 	"Allreduce", ]  - 0
	\caption{Cycles to finish All2All and Allreduce applications in a 2D-HyperX 8x8 network.}%
	\label{fig:label_hamming_2d_perf}%
\end{experimentfigure}

\egroup

\section{Conclusions}
\label{sec:conclusion}
    This paper introduces TERA, a Topology-Embedded Routing Algorithm for Full-mesh networks that enables deadlock-free non-minimal routing without relying on virtual channels.
    An initial evaluation of link ordering schemes exposed their performance limitations, highlighting the need for a more effective solution.

    TERA addresses these challenges by embedding a physical service topology within the Full-mesh structure.
    An exhaustive analysis of the service topologies revealed that the HyperX family, particularly the 2D-HyperX and 3D-HyperX, are the most suitable candidate for TERA.

    In a Full-mesh network, TERA outperforms link ordering routing algorithms by an 80\% under adverse traffic, and up to 100\% in application kernels.
    Furthermore, compared to other VC-based approaches, it reduces buffer requirements by 50\%, while maintaining comparable latency and throughput.
    Lastly, early results from a 2D-HyperX evaluation show that TERA outperforms state-of-the-art algorithms that use the same number of VCs, achieving performance improvements of up to 32\%.

    This work concludes that TERA can be effectively applied to FM networks, and that it is a promising solution for high-radix, low-diameter topologies.

%    Furthermore, preliminary evaluations on a 2D-HyperX topology demonstrate that TERA is well-suited for high-radix, low-diameter topologies built upon Full-mesh cores.
%    It achieves performance improvements of up to 32\% against other state-of-the-art routing algorithms with the same resources.

\clearpage

\appendix
%    \begin{theorem}
%		If a routing scheme is both
%		\begin{itemize}
%		\item based on ordering of the arcs and
%		\item ensures that all links can be used by the same number of source/destination pairs,
%		\end{itemize}
%		then the number of allowed paths of length 2 is $\frac{1}{2}n(n-1)(n-2)$.
%        %Ensuring all the links have the same utilization, the number of non-minimal deadlock-free paths is $\frac{1}{2}\times(n \times (n-1) \times (n-2))$.
%        \label{claim:utilizationfairpaths}
%    \end{theorem}

    \section{Link ordering schemes}\label{app:link_ordering}

    \begin{proof}[Proof of Theorem~\ref{claim:utilizationfairpaths}]
        Let us define the function $\Phi:V^3\rightarrow \{0,1\}$ that indicates which paths of length 2 are allowed. For distinct switches $s,m,d\in V$, set $\Phi(s,m,d)=1$ if the route $s,m,d$ is allowed and set $\Phi(s,m,d)=0$ otherwise. That is, $\Phi$ describes which switches $m$ can be used to route from source $s$ to destination $d$. Additionally, just to avoid cumbersome notation, set $\Phi(s,m,d)=0$ when any pair of $s,m,d$ coincide.
		%%% --- Esto es, ya podemos usar V en vez de V \setminus \{a,b\}.
		%$\Phi (a,b,x)$, where $a,b,x \in V$ and $a \neq b \neq x$ returns 1 if the route $a-b-x$ is valid, and 0 if it is not valid.
        The hypothesis of the same utilization can now be written as follows. There is a constant $S$ such that for any arc $(a, b)$, holds $\sum_{x \in V}  \Phi (a,b,x) + \Phi (x,a,b)=S$. %is equal if and only if the utilization of the links is equal.

        %As all the links are ordered, imagine that the link between $m_a, m_b$ is the least ordered,
		Let $(m_a,m_b)$ be the arc at first position by the ordering of the scheme.
		This arc is always allowed as the start of a 2-length path, but never as its ending. Thus, $\sum_{x \in V } \Phi (m_a, m_b,x) = n-2$, and $\sum_{x \in V} \Phi (x, m_a, m_b) = 0$.
        Therefore, we obtain the value of the constant $S=n-2$.
		%$S=\sum_{x \in E - \{a,b\}}  \Phi (a,b,x) + \Phi (x,a,b)$ = $n-2$ for any link $(a, b)$.

        This implies that
		\begin{equation*}
			n (n-1)  (n-2) =  \sum_{(a,b)\in A} \sum_{x \in V} \Phi (a,b,x) + \Phi (x,a,b).
		\end{equation*}
        Note that the two terms have the same sum, that is,
		\begin{equation*}
			\sum_{(a,b)\in A} \sum_{x \in V} \Phi (a,b,x) = \sum_{(a,b)\in A} \sum_{x \in V}\Phi (x,a,b).
		\end{equation*}
		Therefore,
			$n(n-1)(n-2)=2\sum_{(a,b)\in A} \sum_{x \in V} \Phi (a,b,x)$.
        %$\sum_{(a,b) \in A} \sum_{x \in V - \{a,b\}} \Phi (a,b,x) + \Phi (x,a,b) = 2 \times \sum_{(a,b)}^E \sum_{x \in V - \{a,b\}} \Phi (a,b,x)$.

        %Finally $n \times (n-1) \times (n-2) = 2 \times Allpaths$, so $Allpaths = \frac{1}{2} \times (n \times (n-1) \times (n-2))$
		Finally, the total number of allowed paths is $\sum_{s,m,d\in V}\Phi(s,m,d)=\sum_{a,b\in V} \sum_{x \in V} \Phi (a,b,x)=\frac{1}{2}n(n-1)(n-2)$.
    \end{proof}

%    \begin{claim}
%       In the Forward Distance ordering Scheme, the minimum number of allowed intermediates for source/destination pair in $K_n$ is at least $\frac{n-4}{2}$.
%        \label{claim:intermediates_forward}
%    \end{claim}

    \begin{proof}[Proof of Claim~\ref{claim:intermediates_forward}]

        The number of intermediate switches between a switch source $a$, and a switch destination $b$ is going to be calculated.
        It should be taken into account that the maximum number of intermediates is $n-2$, as $a$ and $b$ do not count as intermediates.

        Let us define the function $G_{ab}(i) := D(i,b) - D(a,i)$.
        A value $G_{ab}(i) > 0$ means that the non-minimal path $a \mapsto i \mapsto b$ is allowed.
        Otherwise, the path is forbidden.

        The function $G$ satisfies
		\begin{equation}\label{eq:forward_pairing}
			-G_{ab}(\frac{a+b}{2}+x) = G_{ab}(\frac{a+b}{2}-x)\;0 \le x < n.
		\end{equation}
        This implies that every valid intermediate in the set of switches is paired with an invalid one.
        Thus, around half of the switches are valid intermediates, and the others are invalid intermediates. We need to discount those cases were both sides of Equation~\ref{eq:forward_pairing} are 0, as both represent invalid intermediates. That is, the number of solutions to $G_{ab}(y)=0$ in the integers modulo $n$.
        %However, cases where $G_{ab}(\frac{a+b}{2}-x) = 0$, both $\frac{a+b}{2}-x$ and $\frac{a+b}{2}+x$ represent invalid intermediates.
        %The solutions of $y = \frac{a+b}{2}$, makes $G_{ab}(y) = 0$.
        %When $n$ is odd, y always has exactly one solution, and there are $\frac{n-3}{2}$ intermediates.
		Taking $x=0$ in Equation~\ref{eq:forward_pairing} gives $G_{ab}(\frac{a+b}{2})=0$.
		Hence, $y=\frac{a+b}{2}$ is a solution when it is  integer.
		More precisely, let us count the solutions in $y$ for $2y\equiv a+b \mod n$.
		When $n$ is odd, 2 has inverse, and we have exactly one solution.
        When $n$ is even:
        \begin{itemize}
            \item if $a$ and $b$ have different parity, there are no solutions, and there are $\frac{n-2}{2}$ intermediates.
            %\item if $a$ and $b$ have the same parity, there is one solution, and there are $\frac{n-3}{2}$ intermediates.
            %\item if $a+b=0$ there are two solutions, and there are $\frac{n-4}{2}$ intermediates.
            \item if $a$ and $b$ have the same parity, there are two solutions, $y=\frac{a+b}{2}$ and $y=\frac{a+b+n}{2}$, and there are $\frac{n-4}{2}$ intermediates.
        \end{itemize}
    \end{proof}

    \section{TERA estimated performance}\label{app:TERA}

    An estimation of the throughput achieved by TERA under various service topologies subjected to Random Switch Permutation (RSP) traffic can be expressed as $\frac{1}{1 + p^{-1}}$, where $p$ denotes the degree of the main topology divided by $n-1$.
    This expression is illustrated for several service topologies in Figure~\ref{fig:safe_links_ratio}.
    The following paragraph outlines the derivation of this expression.

    As the RSP is an adverse traffic pattern, using MIN paths (a single direct link between a pair of switches) is insufficient for achieving good performance.
    Furthermore, the direct link can always be used as a MIN path, regardless of whether it is main or service.
    The 2-hop deroutes must employ a main link in the first hop and any in its second hop.
    Longer routes can be disregarded; they are extremely rare, as it can be empirically seen.
    In consequence, under a reasonable balance of routes, the main links will saturate before the service ones.
    Let $p$ be the ratio of main links to all links, or equivalently, the probability that a randomly chosen link belongs to the main topology.
    Thus, $d=(n-1)p$ is the average degree of the main topology.
    Let $\gamma$ be the average switch injection rate, measured in multiples of full-rate links.
    Assuming traffic without same-switch messages, we have $n-1$ links for all traffic, implying $0\leq \gamma \leq n-1$.
    To estimate $\gamma$, first decompose it into a rate $\gamma_1$ for packets employing 1 hop and a rate $\gamma_2$ for packets employing 2 hops.
    Since most packets do not use them, we can ignore longer routes for an estimation.
    Then, we have that $\gamma=\gamma_1+\gamma_2$.
    Since the traffic is a switch-permutation then $0\leq\gamma_1\leq 1$.
    Packets that traverse a single hop will have the probability $p$ of adding traffic into the main topology.
    For packets traversing two hops, the first one is forced and the second hop only with probability $p$.
    Thus, the $dn/2$ main links, times the two directions, must be able to hold a total of $n(p\gamma_1 + (1+p)\gamma_2)$. This calculation can be
    shown in the following formula:
    \begin{equation}
        n(p\gamma_1 + (1+p)\gamma_2)\leq p(n-1)n.
    \end{equation}
    Note that the equality would be achieved if the load were perfectly distributed. Isolating the major contributor $\gamma_2$ becomes
    \begin{equation}
        \gamma_2\leq \frac{p(n-1)-p\gamma_1}{1+p} = \frac{(n-1)-\gamma_1}{1+p^{-1}}.
    \end{equation}
    Thus, giving a total load per switch of
    \begin{equation}
        \gamma\leq \frac{(n-1)+\gamma_1 p^{-1}}{1+p^{-1}} = \frac{n-1}{1+p^{-1}} + \frac{\gamma_1}{1+p}\leq \frac{n-1}{1+p^{-1}} + 1.
    \end{equation}
    Equivalently, assuming $n$ servers per switch, the average load for each server is $\gamma/n\leq\frac{1}{1+p^{-1}} + O(1/n)$.

    \section*{Acknowledgements}
    This work has been supported by the Spanish Ministry of Science and Innovation under contracts PID2019-105660RB-C22, TED2021-131176B-I00, and PID2022-136454NB-C21.
    C. Camarero is supported by the Spanish Ministry of Science and Innovation, Ramón y Cajal contract RYC2021-033959-I.
    R. Beivide is supported by The Barcelona Supercomputing Center (BSC) under contract CONSER02023011NG.
    Simulations were performed in the Altamira supercomputer, a node of the Spanish Supercomputing Network (RES).

    \bibliography{main}
    \bibliographystyle{plain}

\end{document}